\documentclass[12pt,twoside]{article}
\usepackage[mathscr]{eucal}
\usepackage{amsmath,amsfonts,amssymb,amsthm,mathabx,empheq}
\usepackage{times}
\usepackage{pdfsync}
\usepackage{cite}
\usepackage{url}
\usepackage{hyperref}
\usepackage{tensor}
\usepackage{color}
\usepackage{multicol}
\usepackage{bbold}
\usepackage{tikz}
\usetikzlibrary{calc,decorations.markings}

\advance\textheight\topskip
\allowdisplaybreaks[1]

\newcommand\td{\text{d}}

\newcommand{\p}{\partial}

\newcommand{\be}{\begin{equation}}
\newcommand{\ee}{\end{equation}}
\newcommand{\bea}{\begin{eqnarray}}
\newcommand{\eea}{\end{eqnarray}}
\newtheorem{lemma}{Lemma} [section]
\newtheorem{corollary}{Corollary} [section]

\newcommand*\xbar[1]{%
  \hbox{%
    \vbox{%
      \hrule height 0.5pt 
      \kern0.3ex
      \hbox{%
        \kern-0.0em
        \ensuremath{#1}%
        \kern-0.0em
      }%
    }%
  }%
}

\def\>{\rangle} 
\def\<{\langle} 

\DeclareFontFamily{OT1}{rsfs}{} \DeclareFontShape{OT1}{rsfs}{m}{n}{
<-7> rsfs5 <7-10> rsfs7 <10-> rsfs10}{}
\DeclareMathAlphabet{\mycal}{OT1}{rsfs}{m}{n}

\begin{document}
\numberwithin{equation}{section}
\title{Irrelevant and marginal deformed BMSFT}

\author{xxx}

\date{\today}

\def\mytitle{Irrelevant and marginal deformed BMS field theories}

\addtolength{\headsep}{4pt}

\begin{centering}

  \vspace{1cm}

  \textbf{\Large{\mytitle}}

  \vspace{1.5cm}

  {\large Song He$^{1,2}$, Xin-Cheng Mao$^{1}$ }

\vspace{0.5cm}

\begin{minipage}{.9\textwidth}\small \it  \begin{center}
    $1$ Center for Theoretical Physics and College of Physics,\\ Jilin University, 2699 Qianjin Street,
   Changchun 130012, China
 \end{center}
\end{minipage}

\vspace{0.3cm}

\begin{minipage}{.9\textwidth}\small \it  \begin{center}
     $2$ Max Planck Institute for Gravitational Physics (Albert Einstein Institute),\\
     Am M\"{u}hlenberg 1, 14476 Golm, Germany
 \end{center}
\end{minipage}

\vspace{0.3cm} \vspace{0.3cm}

\begin{minipage}{.9\textwidth}\small \it  \begin{center}
    Email: hesong@jlu.edu.cn, maoxc1120@mails.jlu.edu.cn
 \end{center}
\end{minipage}

\end{centering}

\vspace{1cm}

\begin{center}
\begin{minipage}{.9\textwidth}
  \textsc{Abstract}. 
  In this study, we investigate various deformations within the framework of Bondi-van der Burg-Metzner-Sachs invariant field theory (BMSFT). Specifically, we explore the impact of Bondi-van der Burg-Metzner-Sachs (BMS) symmetry on the theory by introducing key deformations, namely, $T \overline{T}$, $JT_{\mu}$, and $\sqrt{T \overline{T}}$ deformations. In the context of generic seed theories possessing BMS symmetry, we derive the first-order correction of correlation functions using the systematic application of BMS symmetry ward identities. However, it is worth noting that higher-order corrections are intricately dependent on the specific characteristics of the seed theories. To illustrate our findings, we select the BMS free scalar and free fermion as representative seed theories. We then proceed to analytically determine the deformed action by solving the nontrivial flow equations. Additionally, we extend our analysis to include second-order deformations within these deformed theories.

  \end{minipage}
\end{center}
\thispagestyle{empty}

\newpage

\tableofcontents

\newpage

\section{Introduction}

The holographic principle, foundational in quantum gravity \cite{tHooft:1993dmi, Susskind:1994vu}, reshapes our understanding of spacetime, black holes, and the universe, particularly through the anti-de Sitter/conformal field theory (AdS/CFT) correspondence \cite{Maldacena:1997re, Gubser:1998bc, Witten:1998qj}. This correspondence unravels the intricate relationship between quantum gravity properties and boundary field theory. Renormalization group (RG) flow in the boundary field theory provides insights into quantum gravity's behavior across energy scales. We focus on investigating marginal and irrelevant deformations using the double-current construction.

For irrelevant deformations, computational challenges arise due to an infinite number of operators at a fixed point, necessitating an infinite number of counterterms. Solvable classes of 2D spacetime irrelevant deformations, such as $T \overline{T}$ \cite{Zamolodchikov:2004ce, Smirnov:2016lqw, Cavaglia:2016oda} and $J \overline{T}$ \cite{Guica:2017lia} for $U(1)$ current-conserving seed theories, present exceptions. The $JT_{\mu}$ deformation \cite{Anous:2019osb} generalizes $J \overline{T}$. RG flow analysis shows these deformations lead to high-energy scales, disrupting seed theory symmetries. Symmetries of $T \overline{T}$, $J \overline{T}$, and $JT_{\mu}$ deformed CFTs are explored in \cite{Guica:2020uhm}. In addition, $T \overline{T}$ and $JT_{\mu}$ deformations serve as probes for seed theory's ultraviolet (UV) behavior. Marginal deformations, like root-$T \overline{T}$ \cite{Babaei-Aghbolagh:2022uij} (see also \cite{Babaei-Aghbolagh:2022leo,Ferko:2022cix, Rodriguez:2021tcz}), arise from a finite number of operators. Exact marginal deformations maintain symmetries, though perturbative methods introduce deviations, discussed in sections \ref{sec:scalarTTbar} and \ref{sec:deformFermion}. Throughout RG flows, these deformations impact quantum gravity in the bulk through the AdS/CFT correspondence \cite{Witten:2001ua, Klebanov:1999tb, Berkooz:2002ug, Mueck:2002gm, Diaz:2007an, Hartman:2006dy, Gubser:2002zh, Nguyen:2021pdz, Ebert:2023tih}.

To further understand quantum gravity, efforts extend the AdS/CFT correspondence to flat holography \cite{Susskind:1998vk, Polchinski:1999ry,deBoer:2003vf, Arcioni:2003xx, Arcioni:2003td, Solodukhin:2004gs, Barnich:2006av, Guica:2008mu, Barnich:2009se, Barnich:2010eb, Bagchi:2010zz, Bagchi:2012xr}. Recent studies indicate two approaches for flat holography, both centered on the BMS group in asymptotically flat spacetime (AFS) \cite{Bondi:1962px, Sachs:1962wk, Sachs:1962zza}. 

The first flat holography approach, celestial holography, establishes a connection between 4D asymptotically flat spacetime (AFS) quantum gravity and a 2D conformal field theory (CFT) on the celestial sphere at the null boundary $\mathcal{I}^{\pm}$ \cite{Pasterski:2016qvg, Pasterski:2017kqt, Pasterski:2017ylz, Raclariu:2021zjz, Pasterski:2021rjz, Pasterski:2021raf, string}. This celestial CFT incorporates irrelevant or marginal deformations, extending and enriching the celestial holography framework. Co-dimensional two celestial holography applies 2D irrelevant deformations to 4D quantum gravity, potentially constructing UV-complete theories of general relativity \cite{He:2022zcf}. Double current marginal deformations in celestial CFT correspond to loop corrections in 4D AFS scattering amplitudes, impacting the moduli space of bulk vacua \cite{He:2023lvk}, providing precise holographic dictionaries.

The second approach, co-dimension one Carrollian holography, proposes a duality between quantum gravity in AFS and a Carrollian conformal field theory on the null boundary \cite{Dappiaggi:2004kv, Dappiaggi:2005ci}. This duality {can also be observed} through limits in the AdS/CFT correspondence, where the transition from AdS to Minkowski spacetime, i.e., infinite radius, corresponds to the limit of zero light velocity ($c \rightarrow 0$) in the CFT \cite{Bagchi:2012cy, Bagchi:2022nvj, Banerjee:2022ime}. This limit gives rise to a Carrollian CFT or a BMS-invariant field theory (BMSFT). The $c \rightarrow 0$ limit is also known as the ultra-relativistic (UR) limit. The details of re-deriving BMSFT from UR limit are presented in Appendix \ref{sec:URseed}. Further investigations into the connection between Carrollian and flat holography are discussed in \cite{Saha:2023hsl, Nguyen:2023vfz, Salzer:2023jqv, Chen:2023pqf, Chen:2023naw, Bhattacharyya:2023sjr, Bagchi:2023dzx, Bagchi:2023fbj, Saha:2022gjw, Banerjee:2023jpi}.

Given the interconnected nature of various flat holography approaches \cite{Donnay:2022aba, Donnay:2022wvx, Bagchi:2022emh}, our study delves into the impact of irrelevant or marginal deformations, adopting a Carrollian perspective. Correlation functions, pivotal observables in quantum field theories, are the focal points we focused on. Aiming to maintain simplicity and clarity, our examination focuses on the 2D BMSFT as the seed theory, known for its foundational role in flat holography. Specific deformations, such as $T \overline{T}$, $JT_{\mu}$, and $\sqrt{T \overline{T}}$, well-defined in Lorentz-invariant quantum field theories, form the basis of our exploration.

Notably, extending these deformations to the relatively unexplored realm of BMSFT becomes a central objective. Drawing inspiration directly from previous studies \cite{Cardy:2018jho, Cardy:2020olv, Hansen:2020hrs, Chen:2020jdi, Esper:2021hfq, Pavshinkin:2021jpy}, we directly define deformations on BMSFT using the same framework as deformations on 2D Lorentz-invariant quantum field theories. The deformations can also be defined by leveraging the UR limit on Lorentz-invariant quantum field theory definitions. Our proposed approaches to defining deformations within the Carrollian structure are elucidated, with potential equivalence between them suggested by the findings presented in Appendix \ref{sec:DeformUR}. These results offer valuable insights into the compatibility and interchangeability of the two proposed methods, shedding light on the nuanced relationship between deformations and the Carrollian structure in the context of flat holography.

The paper is structured as follows. In Section \ref{sec:Review}, we offer a comprehensive review of the seed theory BMSFT, encompassing an overview of the operator product expansions (OPEs) between the conserved currents and the primary operators. Moreover, we present the non-vanishing correlation functions constructed by these primary operators. Moving on to Section \ref{sec:1stTTbarJTaroot}, we introduce and define the $T \overline{T}, JT_{\mu}, \sqrt{T \overline{T}}$ deformations for 2D BMSFT. We then proceed to perturbatively calculate the first-order correction of these deformations on the correlators in a generic form. To observe the flow effect, achieving accuracy up to at least the second order becomes necessary. Therefore, we apply these deformations to specific examples: the BMS-invariant free scalar model in Section \ref{sec:scalarTTbar} and the BMS-invariant free Fermion model in Section \ref{sec:deformFermion}.  In these two sections, we provide the all-order corrected Lagrangian for the $T \overline{T}, JT_{\mu}$ and $\sqrt{T \overline{T}}$ deformed theory. We then proceed to compute the higher-order corrections of the deformed correlation functions systematically.

\section{2D BMSFT} \label{sec:Review}

\subsection{BMS$_3$ algebra}

The 2D BMSFT is a kind of quantum field theory which is invariant under the following local BMS transform
\begin{equation} \label{eq:BMStran}
    x \rightarrow f(x), \quad y \rightarrow f'(x) y + g(x),
\end{equation}
where $x$ denotes space while $y$ denotes time, and $f(x),g(x)$ are the local dilation and local boost respectively, which can be expanded near $x = 0$ as
\begin{equation}
    f(x) = \sum_{n \in \mathbb{Z}} a_n x^{n+1}, \quad g(x) = \sum_{n \in \mathbb{Z}} b_n x^{n+1}, \quad a_n, b_n \in \mathbb{R}.
\end{equation}
The transform \eqref{eq:BMStran} can be generated by the following BMS generators \cite{Hao:2021urq}
\begin{equation} \label{eq:BMSWitt}
    L_n = - x^n ( x \p_{x} + (n+1) y \p_y), \quad M_n = - x^{n+1} \p_y.
\end{equation}
By central extension, the algebra of BMSFT should be 
\begin{equation} \label{eq:BMFTA}
    \begin{split}
        [L_n, L_m] & = (n-m) L_{n+m} + \frac{c_L}{12} (n^3 - n) \delta_{n+m,0}, \\
        [L_n, M_m] & = (n-m) M_{n+m} + \frac{c_M}{12} (n^3 - n) \delta_{n+m,0}, \\
        [M_n, M_m] & = 0,
    \end{split}
\end{equation}
which refers to $BMS_3$ algebra, equivalently, the 2D Galilean or Carrollian conformal algebra (in short, GCA or CCA). This algebra can be derived from UR limit, see appendix \ref{sec:URAlgebera}. Suppose the action is invariant under transformation \eqref{eq:BMStran}, the components of stress tensor should satisfy \cite{Saha:2022gjw}
\begin{equation}
    T^x_{\ y} = 0, \quad T^{\mu}_{\ \mu} = 0.
\end{equation}
Therefore the components of stress tensor are defined as 
\begin{equation} \label{eq:Tcompo}
    T^{\mu}_{\ \nu} = 
    \begin{pmatrix}
        M & T \\
        0 & -M
    \end{pmatrix},
    \quad T = L + y \p_x M,
\end{equation}
where $T,M$ are Noether current of the translation symmetry along $x,y$ respectively \cite{Hao:2021urq}. Then the conservation law $\p_{\mu} T^{\mu}_{\ \ \nu} = 0$ yields
\begin{equation} \label{eq:conservT}
    \p_y L = 0, \quad \p_y M = 0.
\end{equation}
Therefore the components can be expanded as
\begin{equation} \label{eq:modexpanLM}
    L = \sum_{n \in \mathbb{Z}} L_n x^{-n-2}, \quad M = \sum_{n \in \mathbb{Z}} M_n x^{-n-2}.
\end{equation}
Then by using the algebra \eqref{eq:BMFTA}, the OPEs between the components of stress tensors can be obtained as \footnote{\label{fn:UR} Here is some ambiguity we should remark. While implementing integration, $x$ should be recovered as $z$ by the UR limit. Namely, $x \rightarrow \widetilde{x} = x + i \varepsilon y$, which belongs to the complex plane. Fortunately, this is equivalent to the radial quantization and analytical continuation from cylinder to plane \cite{Hao:2021urq}. That is why we can use \eqref{eq:BMFTA} to derive OPE.}
\begin{equation} \label{eq:T,MOPE}
    \begin{split}
        L(x') L(x) \sim & \frac{c_L}{2(\widetilde{x}'-\widetilde{x})^4} + \frac{2 L(x)}{(\widetilde{x}'-\widetilde{x})^2} + \frac{\p_x L(x)}{\widetilde{x}'-\widetilde{x}}, \quad M(x') M(x) \sim 0,\\
        L(x') M(x) \sim & \frac{c_M}{2(\widetilde{x}'-\widetilde{x})^4} + \frac{2 M(x)}{(\widetilde{x}'-\widetilde{x})^2} + \frac{\p_x M(x)}{\widetilde{x}'-\widetilde{x}}.
    \end{split}
\end{equation}
The $TM$-OPE can be simply obtained by using their relation $T = L + y \p_x M$
\begin{equation} \label{eq:TMOPE}
    \begin{split}
        T(x',y') T(x,y) \sim & \frac{c_L}{2(x'-x)^4} + \frac{2 T(x,y)}{(\widetilde{x}'-\widetilde{x})^2} + \frac{\p_x T(x,y)}{\widetilde{x}'-\widetilde{x}} \\
        & - \frac{2 c_M (y'-y)}{(\widetilde{x}'-\widetilde{x})^5} - \frac{4(y'-y) M(x)}{(\widetilde{x}'-\widetilde{x})^3} - \frac{(y'-y) \p_x M(x)}{(\widetilde{x}'-\widetilde{x})^2}, \\
        T(x',y') M(x) \sim & M(x') T(x,y) \sim \frac{c_M}{2(\widetilde{x}'-\widetilde{x})^4} + \frac{2 M(x)}{(\widetilde{x}'-\widetilde{x})^2} + \frac{\p_x M(x)}{\widetilde{x}'-\widetilde{x}}, \\
        M(x') M(x) \sim & 0.
    \end{split}
\end{equation}


\subsection{Highest weight representation}

In this subsection, we discuss the highest weight representation of 2D BMSFT. Since the BMS algebra \eqref{eq:BMFTA} can be obtained by UR limit from Virasoro algebra, it is straightforward to borrow the highest weight representation from 2D CFT by using UR limit. This representation in BMSFT is referred to as the induced representation \cite{Bagchi:2009ca}, which is unitary. Note, however, that the induced representation is not the highest weight representation of 2D BMSFT, which can be derived parallelly as that of 2D CFT. In this way, the Hilbert space of 2D BMSFT can be decomposed into the BMS module of \eqref{eq:BMFTA} as
\begin{equation}
    \mathcal{H} = \sum_{\Delta, \xi} \mathcal{H}_{\Delta, \xi},
\end{equation}
where $\Delta, \xi$ are eigenvalues of $L_0, M_0$ respectively. The requirement to the primary operator defined on origin $\textbf{O} = \textbf{O} (0,0)$ in each block can be derived by using the state-operator correspondence
\begin{equation}
    [L_n, \textbf{O}] = 0, \quad [M_n, \textbf{O}] = 0, \quad n>0.
\end{equation}
Unfortunately, it turns out that the Kac determinant for the highest weight representation of 2D BMSFT with non-zero boost charge is negative \cite{Bagchi:2009pe}, which indicates that the highest weight representation of 2D BMSFT is not unitary. Therefore, even though $L_0, M_0$ are commutative with each other, they may not be diagonalizable in the same module simultaneously. This will form a novel “multiplet” structure of primary operator which shares a similar feature as logarithmic CFT \cite{Gaberdiel:1996kx, Rohsiepe:1996qj, Gaberdiel:2001tr, Cardy:2013rqg, Creutzig:2013hma, Kytola:2009ax, Hogervorst:2016itc, Gurarie:1993xq}. Then the eigenvalue of $L_0,M_0$ can be read off as
\begin{equation}
    [L_0, O_a] = \Delta O_a, \quad [M_0, O_a] = (\boldsymbol{\xi} O)_a, \quad a = 0, \cdots, r-1
\end{equation}
where the rank $r$ \footnote{In particular, the rank-1 multiplet of primary operators refer to as singlets, which is denoted as $\mathcal{O}$. They form the singlet version of highest weight representation \cite{Bagchi:2009ca}.} is the number of primary operators in the same module that are related to each other, and the matrix $\boldsymbol{\xi}$ can always be chosen as Jordan form
\begin{equation}
    \boldsymbol{\xi} = 
    \begin{pmatrix}
        \xi \\
        1 & \xi \\
        & \ddots & \ddots \\
        & & 1 & \xi
    \end{pmatrix}.
\end{equation}
The operators on arbitrary position can be evolved by $U = e^{x L_{-1} - y M_{-1}}$ as
\begin{equation}
    O_a(x,y) = U O_a(0,0) U^{-1}
\end{equation}
Then the transform of primary operators will be derived by using the Baker-Campbell-Hausdorff (BCH) formula
\begin{align}
    [L_n, O_a (x,y)] = & \left[ x^{n+1} \p_x + (n+1) x^n y \p_y + (n+1) (x^n \Delta + n x^{n-1} y \boldsymbol{\xi}) \right] O_a (x,y), \notag \\
    [M_n, O_a (x,y)] = & \left[ x^{n+1} \p_y + (n+1) x^n \boldsymbol{\xi} \right] O_a (x,y).
\end{align}
Similarly, the OPEs between primary operators and stress tensors can be derived as
\begin{equation} \label{eq:WardM}
    \begin{split}
        T(x',y') O_{a}(x,y) \sim & \frac{\Delta O_{a}}{(\widetilde{x}'-\widetilde{x})^2} + \frac{2(y - y') (\boldsymbol{\xi} \cdot O)_a}{(\widetilde{x}'-\widetilde{x})^3} + \frac{\p_x O_{a}}{\widetilde{x}'-\widetilde{x}} + \frac{(y-y') \p_y O_{a}}{(\widetilde{x}'-\widetilde{x})^2} \\
        M(x') O_{a}(x,y) \sim & \frac{(\boldsymbol{\xi} \cdot O)_a}{(\widetilde{x}'-\widetilde{x})^2} + \frac{\p_y O_{a}}{\widetilde{x}'-\widetilde{x}},
    \end{split}
\end{equation}
where we substituted $\widetilde{x}$ for $x$, see the relevant discussion in footnote \ref{fn:UR}. This matches with the OPEs derived from \cite{Saha:2022gjw}. The OPEs can also be derived from the UR limit, see appendix \ref{sec:UROPE}.

\subsection{2D non-Lorentzian Kac-Moody algebra} \label{sec:NLKM}

The non-vanishing commutators of 2D non-Lorentzian Kac-Moody (NLKM) algebra are \cite{Bagchi:2023dzx}
\begin{equation} \label{eq:NLKM}
    \begin{split}
        [L_n, L_m] & = (m-n) L_{m+n} + \frac{c_L}{12} (m^3 - m) \delta_{m+n, 0}, \\
        [L_m, M_n] & = (m-n) M_{m+n} + \frac{c_M}{12} (m^3 - m) \delta_{m+n, 0}, \\
        [L_m, J_n^a] & = - n J_{m+n}^a, \quad [L_m, K_n^a] = - n K_{m+n}^a, \quad [M_m, J_n^a] = - n K_{m+n}^a, \\
        [J_m^a, J_n^b] & = i F^{abc} J^c_{m+n} + i G^{abc} K_{m+n}^c + m k_1 \delta^{ab} \delta_{m+n, 0}, \\
        [J_m^a, K_n^b] & = i F^{abc} K_{m+n}^c + m k_2 \delta^{ab} \delta_{m+n, 0},
    \end{split}
\end{equation}
where we sum over the double index $c$, and the first two lines are exactly the CCA \eqref{eq:BMFTA} derived in previous subsections, indicating that the subalgebra of NLKM is the $BMS_3$ algebra. The NLKM algebra can also be derived from Virasoro Kac-Moody algebra by taking UR limit, see appendix \ref{sec:URAlgebera}. Furthermore, the NLKM algebra can be intrinsically derived from the conserved Kac-Moody current with the following form \footnote{ Note that the upper index $\mu$ of the current $j^{a \mu}$ is raised by $\epsilon^{\mu \nu}$. Precisely, we have $j^{a \mu} = \epsilon^{\mu \nu} J^a_{\nu}$. This depends on the structure of Newton-Cartan geometry \cite{Hofman:2014loa}, see also the review in \cite{Hao:2022xhq} }
\begin{equation} \label{eq:KMcurrent}
    j^{a \mu} = (J^a_x, - J^a_y),
\end{equation}
where
\begin{equation}
    J_y^a = \sum_{n \in \mathbb Z} x^{-n - 1} K_n^a, \quad J_x^a = \sum_{n \in \mathbb Z} x^{-n -1} \left[ J_n^a - (n+1) \frac{y}{x} K_n^a \right].
\end{equation}
The OPEs between the current and the primary operators are \cite{Bagchi:2023dzx}
\begin{equation} \label{eq:JWard}
    \begin{split}
        J^a_y (x', y') O(x,y) & \sim i \frac{\mathcal{F}^a \cdot O(x,y)}{\widetilde{x}' - \widetilde{x}}, \\
        J^a_x (x', y') O(x,y) & \sim i \frac{\mathcal{G}^a \cdot O(x,y)}{\widetilde{x}' - \widetilde{x}} - i (y' - y) \frac{\mathcal{F}^a \cdot O(x,y)}{(\widetilde{x}' - \widetilde{x})^2},
    \end{split}
\end{equation}
where $\mathcal{G}, \mathcal{F}$ are two independent operators acting on fields, which denote the variation of the fields under the infinitesimal NLKM transformation.

\subsection{Correlators}

\subsubsection{Correlators of singlets}

Since the vacuum is invariant under the global BMS symmetry, the two and three-point function of primary operators can be fixed as
\begin{equation} \label{eq:G2Single}
    \<X_2\> = \<\mathcal{O}_1(x_1,y_1) \mathcal{O}_2(x_2,y_2)\> = N \delta_{\Delta_1, \Delta_2} \delta_{\xi_1, \xi_2} |x_{12}|^{-2 \Delta_1} e^{-2 \xi_1 \frac{y_{12}}{x_{12}}},
\end{equation}
\begin{multline} \label{eq:G3Single}
    \<X_3\> = \<\mathcal{O}_1(x_1,y_1) \mathcal{O}_2(x_2,y_2) \mathcal{O}_3(x_3,y_3) \> \\
    = c_{123} |x_{12}|^{-\Delta_{123}} |x_{23}|^{-\Delta_{231}} |x_{31}|^{-\Delta_{312}} \text{exp} \left( - \xi_{123} \frac{y_{12}}{x_{12}} - \xi_{312} \frac{y_{31}}{x_{31}} - \xi_{231} \frac{y_{23}}{x_{23}} \right),
\end{multline}
where $N$ is the normalization factor, $c_{123}$ encodes dynamical information of BMSFTs, and
\begin{equation}
    x_{ij} = x_i - x_j, \ y_{ij} = y_i - y_j \quad \Delta_{ijk} = \Delta_{i} + \Delta_{j} - \Delta_{k}, \ \xi_{ijk} = \xi_{i} + \xi_{j} - \xi_{k}
\end{equation}
Moreover, the four-point function of primary operators can be defined up to an arbitrary function
\begin{equation} \label{eq:G4Single}
    \begin{split}
        \<X_4\> = & \< \mathcal{O}_1(x_1,y_1) \mathcal{O}_2(x_2,y_2) \mathcal{O}_3(x_3,y_3) \mathcal{O}_4(x_4,y_4) \>\\
        = & \prod_{i<j}^4 |x_{ij}|^{\sum_{k=1}^4 -\Delta_{ijk}/3 } \text{exp} \left( \frac{y_{ij}}{x_{ij}} \sum_{k=1}^4 \frac{\xi_{ijk}}{3} \right) f(\mathcal{X}, \mathcal{Y}),
    \end{split}
\end{equation}
where the following defined cross ratios are invariant under the global BMS transform
\begin{equation}
    \mathcal{X} = \frac{x_{12} x_{34}}{x_{13} x_{24}}, \quad \mathcal{Y} = \frac{y_{12}}{x_{12}} + \frac{y_{34}}{x_{34}} - \frac{y_{13}}{x_{13}} - \frac{y_{24}}{x_{24}}.
\end{equation}

\subsubsection{Correlators of multiplets}

Things will become more complex than singlet primary operators. For two and three functions, we have
\begin{equation} \label{eq:G2Multi}
    \<O_{ia}(x_1, y_1) O_{jb}(x_2, y_2)\> = 
    \begin{cases}
        & 0, \quad q_i = a + b +1 -r_i < 0,\\
        & \delta_{ij} N_i |x_{12}|^{-2 \Delta_i} e^{-2\xi_i \frac{y_{12}}{x_{12}}} \frac{1}{q_i!} \left(-\frac{2 y_{12}}{x_{12}}\right)^{q_i}, \quad q_i \geq 0,
    \end{cases}
\end{equation}
\begin{equation}
    \< O_{ia} O_{jb} O_{kc} \> = A \ B \ C_{ijk}^{abc},
\end{equation}
where $N_i$ is the normalization factor and
\begin{equation} \label{eq:ABC3pt}
    \begin{split}
        A = & \text{exp} \left( - \xi_{ijk} \frac{y_{12}}{x_{12}} - \xi_{kij} \frac{y_{31}}{x_{31}} - \xi_{jki} \frac{y_{23}}{x_{23}} \right), \\
        B = & |x_{12}|^{- \Delta_{123}} |x_{23}|^{- \Delta_{231}} |x_{31}|^{- \Delta_{312}}, \\
        C_{ijk}^{abc} = & \sum_{n_1 = 0}^{a} \sum_{n_2 = 0}^{b} \sum_{n_3 = 0}^{c} c_{ijk}^{n_1n_2n_3} \frac{ p_1^{n_1} p_2^{n_2} p_3^{n_3} }{n_1! n_2! n_3!}, \quad p_i = \p_{\xi_i} \log A.
    \end{split}
\end{equation}

\section{Deformations for the 2D BMSFT} \label{sec:1stTTbarJTaroot}

We will now discuss the effect of irrelevant and marginal deformations on the BMSFT. As we have introduced before, little is known about how to define deformations acting on BMSFT. Therefore, in this section, we will generically implement the definitions of irrelevant and marginal deformations to the seed theory 2D BMSFT. Specifically, we will discuss $T \overline{T}$ and $JT_{\mu}$ for irrelevant deformations and $\sqrt{T \overline{T}}$ for marginal deformation. As we will see, the first-order correction to both correlators and the Lagrangians or actions are all universal and are not affected by the flow of deformations. Our primary concern is where the seed theory will be flowed by these deformations. The flow effect will be reflected in the higher-order corrections, which, unfortunately, are not universal and depend on the concrete seed theory. As a generic introduction to show the universal properties of the deformed correlation functions without knowing the fields, we will mainly focus on the first-order correction of these deformations in this section, while the non-universal, or higher-order corrections will be concerned in the next sections.

\subsection{$T \overline{T}$ deformed BMSFT}

The $T \overline{T}$ deformed action for BMSFT, which is a non-relativistic field theory, can be defined in a similar way as that for CFT, namely
\begin{equation} \label{eq:detT}
    \p_{\lambda} S^{[\lambda]} = \lambda \int \td x \td y O_{T \Bar{T}}^{[\lambda]}, \quad O_{T \Bar{T}} = \text{det} T^{\mu}_{\ \nu}, \quad S^{[\lambda]} [\Phi, \p_{\mu} \Phi] = \int dx dy \mathcal{L}^{[\lambda]}.
\end{equation}
Perturbatively, each quantity can be expanded as a Taylor series by the power of $\lambda$
\begin{equation} \label{eq:Taylor}
    \mathcal{L}^{[\lambda]} = \sum_{n = 1}^{\infty} \frac{\lambda^n}{n!} \mathcal{L}^{(n)}, \quad T^{\mu [\lambda]}_{\ \nu} = \sum_{n = 1}^{\infty} \frac{\lambda^n}{n!} T^{\mu (n)}_{\ \nu}, \quad S^{[\lambda]} = \sum_{n = 1}^{\infty} \frac{\lambda^n}{n!} S^{(n)},
\end{equation}
where $S^{(n)} = \int \td x \td y \mathcal{L}^{(n)}$. Then each order of deformed Lagrangian satisfies the following recursion relation \cite{He:2020cxp}\footnote{One can also refer to the recent systematic investigation \cite{He:2023kgq} of the deformed correlation function which is independent of detailed data of the seed theory.}
\begin{equation} \label{eq:recursion}
    \begin{split}
        \mathcal{L}^{(n+1)} = & \frac{1}{2} \sum^n_{i = 0} C^i_n \left( T^{\mu(i)}_{\ \mu} T^{\nu (n-i)}_{\ \nu} - T^{\mu (i)}_{\ \nu} T^{\nu (n-i)}_{\ \mu} \right), \\
        T^{\mu (n)}_{\ \nu} = & \frac{\p \mathcal{L}^{(n)}}{\p (\p_{\mu} \Phi)} \p_{\nu} \Phi - \delta^{\mu}_{\ \nu} \mathcal{L}^{(n)}, \quad n = 1, 2, \cdots.
    \end{split}
\end{equation}
The $T^{\mu (0)}_{\ \nu}$ is defined in \eqref{eq:Tcompo}. Sometimes $\mathcal{T}^{\mu (0)}_{\ \nu} = \frac{\p \mathcal{L}^{(0)}}{\p (\p_{\mu} \Phi)} \p_{\nu} \Phi - \delta^{\mu}_{\ \nu} \mathcal{L}^{(0)}$ without the EoM of fields is in the same form as $T^{\mu (0)}_{\ \nu}$ defined in \eqref{eq:Tcompo}. In this case, the expression of $T^{\mu (n)}_{\ \nu}$ above can be formally extended to $n = 0, 1, 2, \cdots$. We will show this in section \ref{sec:scalarTTbar}. Generally, without the EoM of fields, $\mathcal{T}^{\mu (0)}_{\ \nu}$ is not in the same form as $T^{\mu (0)}_{\ \nu}$, then the $T^{\mu (n)}_{\ \nu}$ in the above equation should not include $n=0$, and the EoM of fields should be implemented after finishing the computation of all order corrections of the deformed Lagrangian. We will see this in section \ref{sec:deformFermion}. The deformed correlation function can be derived from the path integral definition as
\begin{equation} \label{eq:path<Xn>}
    \<X_n\>_{[\lambda]}^{T \overline{T}} = \frac{ \int \mathcal{D} \Phi \ X_n \ e^{- S^{[\lambda]}} }{ \int \mathcal{D} \Phi e^{- S^{[\lambda]}} } = \frac{ \left\< X_n \ e^{- \delta S} \right\>}{ \left\< e^{- \delta S} \right\>}, \quad \delta S = S^{[\lambda]} - S.
\end{equation}
Therefore, the deformed correlators can be computed order by order
\begin{equation}
    \<X_n\>_{[\lambda]}^{T \overline{T}} = \sum^{\infty}_{n = 0} \frac{\lambda^n}{n!} \<X_n\>^{(n)},
\end{equation}
where
\begin{align}
    \<X_n\>^{(0)}_{T \overline{T}} = & \<X_n\>_{[\lambda = 0]} = \<X_n\> \label{eq:0thTTbar} \\
    \<X_n\>^{(1)}_{T \overline{T}} = & \left\< S^{(1)} \right\> \<X_n\> - \left\< S^{(1)} X_n \right\>, \label{eq:1stTTbar} \\
    \<X_n\>^{(2)}_{T \overline{T}} = & \left\< S^{(1)} S^{(1)} X_n \right\> - \left\< S^{(1)} S^{(1)} \right\> \< X_n \> + \left\< S^{(2)} \right\> \<X_n\> -  \left\< S^{(2)} X_n \right\> \notag \\
    & + 2 \left\< S^{(1)} \right\>^2 \<X_n\> - 2 \left\< S^{(1)} \right\> \left\< S^{(1)} X_n \right\> \label{eq:2ndTTbar} \\
    \vdots & \notag
\end{align}
Note that $\<M^k\> = 0$ for any integral $k>0$, since the $MM$ OPE \eqref{eq:T,MOPE} is zero. Therefore the first-order correction of BMSFT correlators \eqref{eq:1stTTbar} can be derived as follows by using un-deformed stress tensor defined in \eqref{eq:Tcompo} and the recursion relation \eqref{eq:recursion} \footnote{Another way to define the $T\overline{T}$ deformation is to use the UR limit to borrow the definition from the $T \overline{T}$ deformation for Lorentz invariant field theory. At least, we can easily verify that the first-order correction of these two definitions is almost the same, up to a constant, which can be absorbed to the coupling constant $\lambda$ by redefinition. One may find the details in appendix \ref{sec:DeformUR}.}
\begin{equation}
    \<X_n\>^{(1)}_{T \overline{T}} = \int \td x \td y \<(MM)(x) X_n\>, 
\end{equation}
where the integral range of $(x,y)$ are all $(- \infty, \infty)$. Since the Ward identities between $M$ and primary operators are generic, the first-order correction of $T \overline{T}$ deformed correlators is universal. Meanwhile, since the first-order correction is based on the data in seed theory, it does not contain the information on the flow effect, which will appear in the higher-order corrections. However, the higher-order correction will not be universal anymore. We will discuss this later. In this subsection, we only consider the first-order correction to the correlators, which could be derived by using the Ward identity \eqref{eq:WardM} as
\begin{equation} \label{eq:1st}
    \<X_n\>^{(1)}_{T \overline{T}} = \sum_{i,j}^n \int \td x \td y \left[ \frac{\boldsymbol{\xi}_i }{(\widetilde{x}-\widetilde{x}_i)^2} + \frac{\p_{y_i}}{\widetilde{x}-\widetilde{x}_i} \right] \left[ \frac{\boldsymbol{\xi}_j }{(\widetilde{x}-\widetilde{x}_j)^2} + \frac{\p_{y_j}}{\widetilde{x}-\widetilde{x}_j} \right] \<X_n\>.
\end{equation}
Therefore we only need to deal with the integral with the following form
\begin{equation} \label{eq:Integral}
    \mathcal{I}_{a_1 \cdots a_n}^f = \int \td x \td y \frac{f(y-y_i)}{\prod^n_{i=1} (\widetilde{x}-\widetilde{x}_i)^{a_i}},
\end{equation}
where $f(y-y_i)$ is an arbitrary function of the time direction $y$ without poles. This integral can be computed by attaching each operator to an arbitrary operator in $(x_k,y_k), \ k \in \{1, \cdots, n\}$
\begin{equation} \label{eq:If}
    \mathcal{I}_{a_1 \cdots a_n}^f = -\sum_{j = 1}^n \int^{y_k}_{y_j} \td y f(y-y_i) \oint_{x_j} \frac{\td x}{\prod^n_{i=1} (x-x_i)^{a_i}},
\end{equation}
where all $x_k$ are all real numbers while the integral variable $x$ is complex number, see details in appendix \ref{sec:Integral}. Specifically, we will encounter the simplest case $f = 1$ for most computations, such that
\begin{equation} \label{eq:ointall=0}
    \mathcal{I}_{a_1 \cdots a_n} = \sum_{j = 1}^n y_{jk} \oint_{x_j} \frac{\td x}{\prod^n_{i=1} (x-x_i)^{a_i}}.
\end{equation}
With all these preparations, the first-order correction can be easily computed as
\begin{equation} \label{eq:TTbar1st}
    \<X_n\>^{(1)}_{T \overline{T}} = - 2 \pi i \sum_{i \neq j} \frac{y_{ij}}{x_{ij}} \left[ \frac{2}{x_{ij}^2} \boldsymbol{\xi}_i \boldsymbol{\xi}_j + \frac{1}{x_{ij}} \left( \boldsymbol{\xi}_i \p_{y_j} + \boldsymbol{\xi}_j \p_{y_i} \right) + \p_{y_j} \p_{y_i}  \right] \<X_n\>.
\end{equation}
We should remark again that $\boldsymbol{\xi}_i$ is a Jordan matrix acting on the $i$-th field of the correlator $\<X_n\>$ in the generic case, whose size depends on the rank of the multiplet primary fields. Moreover, there still are some derivative operators $\p_{y_i}$, which will have distinguishing behavior depending on the different pole structures of the correlators in the seed theory. The deformation will yield extra pole structures at the first-order correction level for the correlation functions, which depend on the rank of the multiplets and the pole structure of the un-deformed correlators. The extra pole structure might be complex, even though the result \eqref{eq:TTbar1st} seems simple. Therefore, it is worth computing some relevant specific cases to detect these novel structures. To manifest them in the first correction level without knowing the fields themselves, we only need to fix the rank of each field in $\<X_n\>$ and we should know the pole structures of the un-deformed $\<X_n\>$. Specifically, in some cases, the pole structure yielded by the deformed and un-deformed pole structure is simply factorized at the first-order corrected correlators, which will appear only when the pole structure can be fixed by the BMS symmetries in the seed theory and the rank of the fields are all 1. Fortunately, the singlet version of 2-point and 3-point functions satisfy these two conditions. But the first-order correction of 4-point functions may not be factorized, since the un-deformed pole structure of 4-point functions is not completely fixed, up to an arbitrary function of cross-ratio. Next in this subsection, we will compute 2-point and 3-point functions for the singlet version to show the factorization property. As a comparison, we will compute the 4-point functions in the singlet version and correlators at 2-point, and 3-point to show that they are not factorized.

\subsubsection{Correction to singlets}

In the singlet case, the rank of the matrix $\boldsymbol{\xi}$ in \eqref{eq:TTbar1st} is 1, namely $\boldsymbol{\xi} = \xi$ is a number, rather than a matrix. 

\paragraph{2-point} From \eqref{eq:G2Single}, we notice that only operators that have the same weights are non-zero. Thus, the non-vanishing singlet 2-point function in seed theory is constructed by two same singlets $\mathcal{O}$ with the conformal dimension $\Delta$ and boost charge $\xi$ defined at different points
\begin{equation}
    \<X_2\> = \<\mathcal{O} (x_1,y_1) \mathcal{O} (x_2,y_2)\> = N |x_{12}|^{-2 \Delta} e^{-2 \xi \frac{y_{12}}{x_{12}}}.
\end{equation}
Then the first-order correction is
Then the first-order correction through $T \overline{T}$ flow is
\begin{equation}
     \< X_2 \>^{(1)}_{T \overline{T}} = - 40 \pi i \xi^2 \frac{y_{12}}{x_{12}^3} \< X_2 \>,
\end{equation}
which is exactly factorized. Note that $x_k$ are real numbers since we only implement analytical continuation for $x$, not $x_k$. The normalization factor $N$ can absorb the probable minus sign caused by the removal of the absolute value sign of $x_{12}$.

\paragraph{3-point} 
We only consider a specific case in 3-pt case: three operators with the same $(\Delta, \xi)$ are placed on 3 different points $(x_1,y_1), (x_2,y_2), (x_3,y_3)$. Namely, we choose
\begin{equation}
    \mathcal{O}_1 = \mathcal{O}_2 =\mathcal{O}_3 = \mathcal{O}, \quad \xi_1 = \xi_2 = \xi_3 = \xi, \quad \Delta_1 = \Delta_2 = \Delta_3 = \Delta.
\end{equation}
So the un-deformed 3-pt \eqref{eq:G3Single} is
\begin{equation}
    \<X_3\> = \< \mathcal{O}(x_1,y_1) \mathcal{O}(x_2,y_2) \mathcal{O}(x_3,y_3) \> = c_{123} (x_{12} x_{23} x_{13})^{-\Delta} e^{ - \xi \left( \frac{y_{12}}{x_{12}} + \frac{y_{31}}{x_{31}} + \frac{y_{23}}{x_{23}} \right) },
\end{equation}
where we absorbed the probable minus sign into $c_{123}$. 
Therefore the first-order correction 
is
\begin{align}
    \< X_3 \>^{(1)}_{T \overline{T}} = 20 \pi i \xi^2 \left( \frac{y_{12}}{x_{12}^3} + \frac{y_{13}}{x_{13}^3} + \frac{y_{23}}{x_{23}^3} \right) \< X_3 \>.
\end{align}
where we used \eqref{eq:ointall=0} and the residue theorem. We can easily see that this result is factorized.

\paragraph{4-point} We first consider 4 same operators in four different points. In this case
\begin{equation}
    \xi_2 = \xi_4 = \xi_1 = \xi_3 = \xi, \quad \Delta_1 = \Delta_2 = \Delta_3 = \Delta_4 = \Delta
\end{equation}
then the un-deformed 4-pt \eqref{eq:G4Single} is
\begin{equation}
    \<\mathcal{O}(x_1,y_1) \mathcal{O}(x_2,y_2) \mathcal{O}(x_3,y_3) \mathcal{O}(x_4,y_4)\> = \frac{F(\mathcal{X}, \mathcal{Y})}{|x_{13} x_{24}|^{2 \Delta}} \text{exp} \left[ - 2 \xi \left( \frac{y_{13}}{x_{13}} + \frac{y_{24}}{x_{24}} \right) \right],
\end{equation}
where we used the following formulas to simplify some of the ratios and absorb extra $\mathcal{X},\mathcal{Y}$ into $F(\mathcal{X},\mathcal{Y})$
\begin{equation*}
    \begin{split}
        x_{12} x_{34} = x_{13} x_{24} \mathcal{X}, \quad & x_{14} x_{23} = x_{13} x_{24} (1 - \mathcal{X}), \\
        \frac{y_{23}}{x_{23}} + \frac{y_{14}}{x_{14}} = \frac{\mathcal{Y}}{\mathcal{X}-1} + \frac{y_{13}}{x_{13}} + \frac{y_{24}}{x_{24}}, & \quad \frac{y_{12}}{x_{12}} + \frac{y_{34}}{x_{34}} = \frac{\mathcal{Y}}{\mathcal{X}} + \frac{y_{13}}{x_{13}} + \frac{y_{24}}{x_{24}}.
    \end{split}
\end{equation*}
Then the 1-st correction of this case is
\begin{align}
    & \<\mathcal{O}(x_1,y_1) \mathcal{O}(x_2,y_2) \mathcal{O}(x_3,y_3) \mathcal{O}(x_4,y_4)\>^{(1)}_{T \overline{T}} \notag \\
    = & - \bigg[ \xi^2 ( x_{24}^4 \mathcal{I}_{0404} + x_{13}^4 \mathcal{I}_{4040} ) + x_{14}^2 x_{23}^2 \left( \mathcal{X}^2 ( \p_{\mathcal{Y}}^2 \text{ln} F + (\p_{\mathcal{Y}} \text{ln} F)^2 ) + \frac{2 \xi^2}{(1-\mathcal{X})^2} \right) \mathcal{I}_{2222} \notag \\
    & + 2 \xi x_{14} x_{23} \mathcal{X} \p_{\mathcal{Y}} \text{ln} F \left( x_{24}^2 \mathcal{I}_{1313} + x_{13}^2 \mathcal{I}_{3131} \right) \bigg] \frac{F(\mathcal{X}, \mathcal{Y})}{|x_{13} x_{24}|^{2 \Delta}} e^{ - 2 \xi \left( \frac{y_{13}}{x_{13}} + \frac{y_{24}}{x_{24}} \right) }, \label{eq:OOOOsingle}
\end{align}
see the integrals in appendix \ref{sec:Integral}.

Then we consider another case: put $\mathcal{O}$ on 2,4 and put $\mathcal{O}^{\dagger}$ on 1,3. In this case, we have
\begin{equation}
    \xi_2 = \xi_4 = - \xi_1 = - \xi_3 = \xi, \quad \Delta_1 = \Delta_2 = \Delta_3 = \Delta_4 = \Delta.
\end{equation}
So the un-deformed 4-point function is \eqref{eq:G4Single}
\begin{equation}
    \<X'_4\> = \<\mathcal{O}^{\dagger}(x_1,y_1) \mathcal{O}(x_2,y_2) \mathcal{O}^{\dagger}(x_3,y_3) \mathcal{O}(x_4,y_4)\> = \frac{F(\mathcal{X}, \mathcal{Y})}{|x_{13} x_{24}|^{2 \Delta}} \text{exp} \left[ \frac{2 \xi}{3} \left( \frac{y_{24}}{x_{24}} - \frac{y_{13}}{x_{13}} \right) \right]
\end{equation}
Then the 1-st correction  of the 4-point function in this case is
\begin{equation} \label{eq:O+OO+Osingle}
    \begin{split}
        & \<\mathcal{O}^{\dagger}(x_1,y_1) \mathcal{O}(x_2,y_2) \mathcal{O}^{\dagger}(x_3,y_3) \mathcal{O}(x_4,y_4)\>^{(1)}_{T \overline{T}} \\
        = & - \bigg\{ x_{14}^2 x_{23}^2  \mathcal{I}_{2222} \mathcal{X}^2 \Big[ \p_{\mathcal{Y}}^2 \text{ln} F + (\p_{\mathcal{Y}} \text{ln} F)^2 \Big] + \xi^2 \bigg[ x_{24}^4 \mathcal{I}_{0404} + x_{13}^4 \mathcal{I}_{4040} - 2 x_{24}^2 x_{13}^2 \mathcal{I}_{2222} \\
        & + \frac{64}{9} (\mathcal{I}_{0202} + \mathcal{I}_{2020} - 2 \mathcal{I}_{1111}) + \frac{16}{3} \Big( x_{24}^2 (\mathcal{I}_{0303} - \mathcal{I}_{1212}) + x_{13}^2 (\mathcal{I}_{3030} - \mathcal{I}_{2121})\Big) \bigg] \\
        & + 2 \xi x_{14} x_{23} \mathcal{X} \p_{\mathcal{Y}} \text{ln} F \left[ x_{24}^2 \mathcal{I}_{1313} - x_{13}^2 \mathcal{I}_{3131} + \frac{8}{3} ( \mathcal{I}_{1212} - \mathcal{I}_{2121} ) \right] \bigg\} \frac{F(\mathcal{X}, \mathcal{Y})}{|x_{13} x_{24}|^{2 \Delta}} e^{ \frac{2 \xi}{3} \left( \frac{y_{24}}{x_{24}} - \frac{y_{13}}{x_{13}} \right) },
    \end{split}
\end{equation}
see the integrals in appendix \ref{sec:Integral}. In general, we can explicitly see that the corrections of 4-point functions in the first-order level is not factorized.

\subsubsection{Correction to multiplets}

Then we will show the first-order correction to the 2-point and 3-point functions composed with multiple operators and show that are not factorized in the first-order level.

\paragraph{2-point} From \eqref{eq:G2Multi} we notice that $O_{ia}$ and $O_{jb}$ must in the same multiplet, namely $i = j$. So we can drop $i,j$, put them in 1,2 respectively ($a + b + 1 - r \geq 0$)
\begin{equation}
    \<O_{a}(x_1, y_1) O_{b}(x_2, y_2)\> = \frac{N |x_{12}|^{ -2 \Delta } e^{-2 \xi \frac{y_{12}}{x_{12}}}}{(a+b+1-r)!} \left(-\frac{2 y_{12}}{x_{12}}\right)^{a+b+1-r}.
\end{equation}
For $r = 1$, the 1-st correction degenerates to the result of the singlet. We only discuss the following condition
\begin{equation}
    a \geq 2, \quad b \geq 2, \quad a + b + 1 - r \geq 0.
\end{equation}
The MM-insertion as
\begin{multline}
    \<MM(x) O_a (x_1,y_1) O_b (x_2,y_2)\> = \bigg\{ ( Q_{a+b}^2 - P_{a+b} ) \left[ \frac{1}{(x-x_1)^4} + \frac{1}{(x-x_2)^4} \right] + \frac{6 (Q_{a+b}^2-P_{a+b})}{(x-x_1)^2 (x-x_2)^2} \\
    - \frac{2 (Q_{a+b}^2 - P_{a+b} - 2 Q_{a+b} \xi)}{(x-x_1) (x-x_2)} \left[ \frac{1}{(x-x_1)^2} + \frac{1}{(x-x_2)^2} \right] \bigg\} \<O_a O_b\>,
\end{multline}
where
\begin{equation}
    \begin{split}
        P_{a+b} = & \frac{x_{12}}{2y_{12}} (Q_{a+b} + \xi) = (a+b+1-r) \left( \frac{x_{12}}{2y_{12}} \right)^2, \\
        Q_{a+b} = & (a+b+1-r) \frac{x_{12}}{2y_{12}} - \xi
    \end{split}
\end{equation}
Then, the 1-st correction of the multiplet 2-point function is
\begin{align}
    \<O_a(x_1,y_1) O_b(x_2,y_2)\>^{(1)}_{T \overline{T}} 
    = & - 16 \pi i \frac{y_{12}}{x_{12}^3} ( 2 Q_{a+b}^2 - 2 P_{a+b} - Q_{a+b} \xi ) \<O_aO_b\>. \label{eq:multi2pt}
\end{align}

\paragraph{3-point} consider $a,b,c \geq 2$ and $i=j=k$ (so we can drop $ijk$ for simplicity). Put $O_a,O_b,O_c$ on 1,2,3 respectively. 
Then, the 1-st correction of the 3-point multiplet is
\begin{equation}
    \begin{split}
        & \< O_{a}(x_1,y_1) O_{b}(x_2,y_2) O_{c}(x_3,y_3) \>^{(1)}_{T \overline{T}} \\
        = & 2 \pi i B \bigg\{ 2 \frac{y_{12}}{x_{12}} \p_{y_1} \p_{y_2} (C^{abc} A) + 2 \frac{y_{13}}{x_{13}} \p_{y_1} \p_{y_3} (C^{abc} A) + 2 \frac{y_{23}}{x_{23}} \p_{y_2} \p_{y_3} (C^{abc} A) \\
        & - 2 \xi \left[ \left( \frac{y_{21}}{x_{21}^2} + \frac{y_{31}}{x_{31}^2} \right) \p_{y_1} (C^{abc} A) + \left( \frac{y_{12}}{x_{12}^2} + \frac{y_{32}}{x_{32}^2} \right) \p_{y_2} (C^{abc} A) + \left( \frac{y_{13}}{x_{13}^2} + \frac{y_{23}}{x_{23}^2} \right) \p_{y_3} (C^{abc} A) \right] \\
        & -2 A \left( \frac{y_{12}}{x_{12}^3} C^{a-1,b-1,c} + \frac{y_{13}}{x_{13}^3} C^{a-1,b,c-1} A + \frac{y_{23}}{x_{23}^3} C^{a,b-1,c-1} \right) \\
        & - 2 A \xi \left[ \left( \frac{y_{12}}{x_{12}^3} + \frac{y_{13}}{x_{13}^3} \right) C^{a-1,b,c} + \left( \frac{y_{12}}{x_{12}^3} + \frac{y_{23}}{x_{23}^3} \right) C^{a,b-1,c} + \left( \frac{y_{13}}{x_{13}^3} + \frac{y_{23}}{x_{23}^3} \right) C^{a,b,c-1} \right] \\
        & - \frac{y_{21}}{x_{21}^2} \p_{y_1} (C^{a,b-1,c} A) - \frac{y_{31}}{x_{31}^2}  \p_{y_1} (C^{a,b,c-1} A)  - \frac{y_{12}}{x_{12}^2} \p_{y_2} (C^{a-1,b,c} A) - \frac{y_{32}}{x_{32}^2} \p_{y_2} (C^{a,b,c-1} A) \\
        & - \frac{y_{13}}{x_{13}^2} \p_{y_3} (C^{a-1,b,c} A) - \frac{y_{23}}{x_{23}^2} \p_{y_3} (C^{a,b-1,c} A) - 4 \xi^2 A C^{abc} \left( \frac{y_{12}}{x_{12}^3} + \frac{y_{13}}{x_{13}^3} + \frac{y_{23}}{x_{23}^3} \right) \bigg\},
    \end{split}
\end{equation}
which may not be factorized.

From now on, we have shown the first-order corrections to the correlation functions, which do not depend on the seed theory. As a perturbative version, the higher-order corrections are based on the data defined in lower orders, which makes the higher-order corrections non-universal and thus depend on the seed theory. The higher-order corrections will be seen in the next sections, which discuss the concrete examples of $T \overline{T}$ deformed free scalar and free Fermion models.

\subsection{$JT_{\mu}$ deformed BMSFT} \label{sec:JTa1st}

The $JT_{\mu}$ deformation can be implemented for the seed BMSFT which contains the NLKM symmetries. The generic definition of $JT_{\mu}$ deformation can be borrowed from \cite{Anous:2019osb} as
\begin{equation}
    \frac{\p S^{[\lambda]}}{\p \lambda^{\mu}_a} = \int \td x \td y \ \epsilon_{\alpha \beta} \ j^{a \alpha}_{[\lambda]} T^{\beta [\lambda]}_{\ \mu}
\end{equation}
where $j^{a \mu}$ is the Kac-Moody current. The first-order correction to the action is also universal, which can be expressed as
\begin{align}
    S^{[\lambda]} = & S^{[0]} + \lambda_a^0 \int \td x \td y \left( j^{a y} T^{x}_{\ y} - j^{ax} T^{y}_{\ y} \right) + \lambda^1_a \int \td x \td y \left( j^{a y} T^{x}_{\ x} - j^{a x} T^{y}_{\ x} \right) + o (\lambda^{\mu}_a) \notag \\
    = & S^{[0]} + \lambda_a^0 \int \td x \td y J^{a}_y M + \lambda^1_a \int \td x \td y \left( J^{a}_y T - J^{a}_x M \right) + o (\lambda^{\mu}_a)
\end{align}
where the current without “$[\lambda]$” is the data of the seed theory, and we used \eqref{eq:Tcompo} and \eqref{eq:KMcurrent}. Through the path integral, the first-order correction to the correlation function will be
\begin{equation}
    \<X_n \>_{[\lambda]}^{JT_{\mu}} = \<X_n\>^{(0)}_{JT_{\mu}} + \lambda_a^{\mu} \<X_n \>_{\mu}^{a (1)} + o(\lambda_a^{\mu})
\end{equation}
where
\begin{equation}
    \begin{split}
        \<X_n\>^{(0)} & = \<X_n\>, \\
        \<X_n \>_0^{a (1)} & = - \int \td x \td y \< J^{a}_y M X_n \>, \\
        \<X_n \>_1^{a (1)} & = - \int \td x \td y \left( \< J^{a}_y T X_n\> - \< J^{a}_x M X_n \> \right).
    \end{split}
\end{equation}
We can use the Ward identities \eqref{eq:WardM} and \eqref{eq:JWard} to compute the first-order corrections for a generic $n$-point correlation function
\begin{equation}
    \<X_n\>^{a(1)}_0 = i \sum_{i,j = 1}^n \int \td x \td y \frac{\mathcal{F}^a_i}{x - x_i} \left[ \frac{\boldsymbol{\xi}_j}{(x - x_j)^2} + \frac{\p_{y_j}}{x-x_j} \right] \<X_n\>
\end{equation}
\begin{equation}
    \begin{split}
        \<X_n\>^{a(1)}_1 & = i \sum_{i,j = 1}^n \int \td x \td y \bigg\{ \left[ \frac{\mathcal{G}^a_i}{x - x_i} - \frac{y - y_i}{(x - x_i)^2} \mathcal{F}^a_i \right] \left[ \frac{\boldsymbol{\xi}_j}{(x - x_j)^2} + \frac{\p_{y_j}}{x-x_j} \right] \\
        & - \frac{\mathcal{F}^a_i}{x - x_i} \left[ \frac{\Delta_j}{(x-x_j)^2} - \frac{2(y - y_j)}{(x-x_j)^3} \boldsymbol{\xi}_j + \frac{\p_{x_j}}{x-x_j} - \frac{(y-y_j) \p_{y_j}}{(x-x_j)^2} \right] \bigg\} \<X_n\>
    \end{split}
\end{equation}
By using the integral scheme \eqref{eq:If}, one obtains
\begin{equation}
    \<X_n\>^{a(1)}_0 = - 2 \pi \sum_{i \neq j} \frac{y_{ij}}{x_{ij}^2} \mathcal{F}^a_i (\boldsymbol{\xi}_j + x_{ij} \p_{y_j} ) \<X_n\>,
\end{equation}
and
\begin{equation}
    \<X_n\>^{a(1)}_1 = 2 \pi \sum_{i \neq j} \frac{y_{ij}}{x_{ij}} \left[ \mathcal{F}^a_i \left( \frac{\Delta_j}{x_{ij}} + \frac{y_{ij}}{2 x_{ij}} \p_{y_j} + \p_{x_j} \right) - \mathcal{G}^a_i \left( \frac{\boldsymbol{\xi}_j}{x_{ij}} + \p_{y_j} \right) \right] \<X_n\>.
\end{equation}
Similar to the discussion of $T \overline{T}$ case, these generic and simple results do not completely manifest the pole structure of the deformed correlation functions at the first-order correction level. Unfortunately, even though the un-deformed pole structure and the rank are fixed, the extra pole structure yielded by the deformation still cannot be completely displayed without knowing the fields themselves because $\mathcal{F}^a_i$ depend on the internal structure of the $i$-th fields, which is different for distinct fields. Then there is no need for this subsection to discuss examples like 2-point and 3-point functions for $JT_{\mu}$ deformations, which will be left for the deformed free scalar and Fermion models.

\subsection{Root-$T \overline{T}$ deformed BMSFT} \label{sec:1strootTTbar}

The $\sqrt{T \overline{T}}$ deformation is defined as \cite{Ferko:2022cix} 
\begin{equation} \label{eq:rootdefine}
    \p_{\lambda} S^{[\lambda]} = \int \td x \td y R^{[\lambda]},
\end{equation}
where $\lambda$ is a dimensionless coupling constant, and $R^{[\lambda]}$ is defined as
\begin{equation} \label{eq:R}
    R^{[\lambda]} = \sqrt{\frac{1}{2} T^{A [\lambda]}_{\ B} T^{B [\lambda]}_{\ A} - \frac{1}{4} \left(T^{A [\lambda]}_{\ A} \right)^2} 
\end{equation}
Similarly, quantities like the stress tensor, the Lagrangian, and the action can be expanded as \eqref{eq:Taylor}. So the recursion relation can be derived from the following formula
\begin{equation} \label{eq:rootrecursion}
    \begin{split}
        \sum_{n = 0}^{\infty} \frac{\lambda^n}{n!} \mathcal{L}^{(n+1)} & = R^{[\lambda]} = \sqrt{\sum_{n = 0}^{\infty} \frac{\lambda^n}{n!} \sum_{i = 0}^n C_n^i \left( \frac{1}{2} T^{A (i)}_{\ B} T^{B (n-i)}_{\ A} - \frac{1}{4} T^{A (i)}_{\ A} T^{B (n - i)}_{\ B} \right)} 
    \end{split}
\end{equation}
The explicit form of $\mathcal{L}^{(n+1)}$ cannot be presented easily because we need to expand the square root around $\text{det} [T^{\mu (0)}_{\ \nu}]$ by the power of $\lambda$. Therefore the recursion relation in $\sqrt{T \overline{T}}$ case is not as simple as that in $T \overline{T}$ case \eqref{eq:recursion}. But the first-order correction to the Lagrangian is still universal
\begin{equation}
    \mathcal{L}^{(1)} = R^{(0)} = M.
\end{equation}
The corrections of correlators will have the same form as \eqref{eq:0thTTbar} \eqref{eq:1stTTbar} and \eqref{eq:2ndTTbar} with different $S^{(n)}$-s.

The first-order corrected correlator can be computed by using \eqref{eq:1stTTbar}
\begin{equation}
    \< X_n \>^{(1)}_{\sqrt{T \overline{T}}} = - \int \td x \td y \<M(x) X_n\> = - \sum_{k = 1}^n \int \td x \td y \frac{\p_{y_k}}{x-x_k} \<X_n\>.
\end{equation}
where the term associated with the boost charge $\frac{(\boldsymbol{\xi}_i \cdot O_i)_a}{(\Delta \Tilde{x}_i)^2}$ or $\frac{\xi_i \mathcal{O}_i}{(\Delta \Tilde{x}_i)^2}$ are zero by using the residue theorem. Since \eqref{eq:oint = 0} will be spoiled in this case, we cannot use the integral scheme \eqref{eq:attachk} directly. We need to divide the integral into 2 parts: $y>y_k, y<y_k$, but the contour of $x$ only contains half of the plane, namely, we can drop one range of $y$. Dropping $y>y_k$ or $y<y_k$ are equivalent, we will see the reason as follows. Firstly, we drop $y<y_k$ to compute contour of $x$ surrounding upper half plane the 1-st order correction as
\begin{equation}
    \< X_n \>^{(1)}_{\sqrt{T \overline{T}}} = - \sum_k \lim_{\Lambda \rightarrow \infty} \int^{\Lambda}_{y_k} dy \oint_{x_k} dx \frac{\p_{y_k}}{x-x_k} \<X_n\> = - 2 \pi i \sum_k (\Lambda - y_k) \p_{y_k} \<X_n\>
\end{equation}
Then, by using the translation conservation, namely $\sum_k \p_{y_k} = 0$, we obtain
\begin{equation} \label{eq:1stroot}
    \< X_n \>^{(1)}_{\sqrt{T \overline{T}}} = 2 \pi i \sum_k y_k \p_{y_k} \<X_n\>
\end{equation}
Next we drop $y>y_k$
\begin{equation}
    \< X_n \>^{(1)}_{\sqrt{T \overline{T}}} = \sum_k \lim_{\Lambda \rightarrow \infty} \int_{-\Lambda}^{y_k} dy \oint_{x_k} dx \frac{\p_{y_k}}{x-x_k} \<X_n\> = 2 \pi i \sum_k (\Lambda + y_k) \p_{y_k} \<X_n\>
\end{equation}
where the lower half-plane contour of $x$ has a minus sign difference from the upper half-plane. By using conservation law to omit $\Lambda$, we can similarly obtain \eqref{eq:1stroot}. As a final remark, the generic form of 1-st order corrected root-$T \overline{T}$ deformation is \eqref{eq:1stroot} no matter whether the un-deformed correlator is multiplet or singlet. The higher-order corrections are not universal and depend on different theories. Next, we will implement the $\sqrt{T \overline{T}}$ deformation to the free scalar model to see the higher-order effect.

\section{Deformed BMS free scalar model} \label{sec:scalarTTbar}


\subsection{Data of seed theory}

This subsection gives the data of the seed theory, the BMS free scalar model, to be well prepared for the deformation. We mainly review \cite{Hao:2021urq} here. The un-deformed Lagrangian of the BMSFT scalar model is
\begin{equation} \label{eq:ScalarL}
    \mathcal{L}^{(0)} = (\p_y \phi)^2
\end{equation}
The equation of motion is
\begin{equation} \label{eq:scalarEoM}
    \p_y^2 \phi = 0.
\end{equation}
There are three kinds of primary operators defined in the seed theory: the identity operator, a rank-2 multiplet
\begin{equation}
    O_0 (x) = i \p_y \phi, \quad O_1 (x,y) = i \p_x \phi,
\end{equation}
and a singlet vertex operator
\begin{equation}
    V_{\alpha} = : e^{\alpha \phi (x,y)} :.
\end{equation}
The components of stress tensor defined in the seed theory can be derived from the definition \eqref{eq:recursion} as
\begin{equation} \label{eq:T0BMSFT}
    T^{y (0)}_{\ x} = T = 2 \p_y \phi \p_x \phi, \quad T^{y (0)}_{\ y} = M = - T^{x (0)}_{\ x} = (\p_y \phi)^2, \quad T^{x (0)}_{\ y} = 0,
\end{equation}
which are consist with the generic discussion \eqref{eq:Tcompo}. The stress tensor equipped with the EoM of fields \eqref{eq:scalarEoM} is conserved, satisfying \eqref{eq:conservT}. The OPE between two scalar fields defined at different points is
\begin{equation} \label{eq:phiphiOPE}
    \phi(x_1, y_1) \phi(x_2, y_2) \sim - \frac{y_{12}}{x_{12}}.
\end{equation}
One can easily use the above OPE to derive the OPEs between primary operators and stress tensors
\begin{equation} \label{eq:TOOPE}
    \begin{split}
        T(x',y') O_0(x) \sim & \frac{O_0}{(x'-x)^2} + \frac{\p_x O_0}{x'-x}, \\
        T(x',y') O_1(x,y) \sim & \frac{O_1}{(x'-x)^2} + \frac{\p_x O_1}{x'-x} - \frac{2(y'-y)O_0}{(x'-x)^3} - \frac{y'-y}{(x'-x)^2} \p_y O_1, \\
        M(x') O_0 (x) \sim & 0, \quad M(x') O_1(x,y) \sim \frac{O_0 (x)}{(x'-x)^2} + \frac{\p_y O_1 (x,y)}{x'-x},
    \end{split}
\end{equation}
and
\begin{equation} \label{eq:TVOPE}
    \begin{split}
        T(x',y') V_{\alpha} (x,y) \sim & \frac{\p_y V_{\alpha} (x,y)}{x'-x} - \frac{(y'-y) \p_y V_{\alpha} (x,y)}{(x'-x)^2} + \frac{\alpha^2 (y'-y) V_{\alpha} (x,y)}{(x'-x)^3}, \\
        M(x') V_{\alpha} (x,y) \sim & \frac{\p_y V_{\alpha} (x,y)}{x'-x} - \frac{\alpha^2 V_{\alpha} (x,y)}{2 (x'-x)^2},
    \end{split}
\end{equation}
Comparing with the generic form \eqref{eq:WardM}, $O_a$ is a multiplet with weight $\Delta = 1$ and vanishing $\xi = 0$, while $V_{\alpha}$ is a singlet with boost charge $\xi = - \frac{\alpha^2}{2}$ and vanishing weight $\Delta = 0$. The OPEs between the components of the stress tensor can also be derived from the contraction \eqref{eq:phiphiOPE}, which are also verified to be consistent with the generic form \eqref{eq:TMOPE} with $c_L = 2$ and $c_M = 0$. 

Note that the seed Lagrangian \eqref{eq:ScalarL} also has the following affine $U(1)$ symmetry
\begin{equation}
    \phi \rightarrow \phi' (x,y) = \phi (x,y) + \Lambda (x).
\end{equation}
Correspondingly, the Noether current in the seed theory is
\begin{equation}
    j^{\mu}_{(0)} = - \frac{\p \mathcal{L}^{(0)}}{\p (\p_{\mu} \phi)} = (J, 0), \quad j^y_{(0)} = J_x = J = -2 \p_y \phi.
\end{equation}
The OPEs between $J$ and primary operators in the seed theory are
\begin{equation} \label{eq:JOOPE}
    \begin{split}
        J (x') O_0 (x,y) & \sim 0, \\
        J (x') O_1 (x,y) & \sim \frac{-i}{x' - x}, \\
        J (x') V_{\alpha} (x,y) & \sim \frac{- \alpha}{x' - x} V_{\alpha} (x,y),
    \end{split}
\end{equation}
which means that the $\mathcal{F}^a_i, \mathcal{G}^a_i$ defined in subsection \ref{sec:NLKM} are 
\begin{equation} \label{eq:G}
    \mathcal{G}_0 = 0, \quad \mathcal{G}_1 O_1 = - 1, \quad \mathcal{G}_V = i \alpha, \quad \mathcal{F}_i = 0.
\end{equation}
for $O_0, O_1$ and $V_{\alpha}$ respectively. Here we dropped the superscript $a$, since the current $J$ itself has no group indices.

The correlators among the primary operators can be derived from the OPEs between the fields. We present the non-vanishing correlators here
\begin{equation} \label{eq:Xscalar0}
    \begin{split}
        & \<O_0(x_1) O_1(x_2,y_2)\>^{(0)} = \frac{1}{x_{12}^2}, \quad \<O_1(x_1,y_1) O_1(x_2,y_2)\>^{(0)} = - \frac{2y_{12}}{x_{12}^3}, \\
        & \<O_0(x_1) V_{\alpha}(x_2,y_2) V_{-\alpha}(x_3,y_3)\>^{(0)} = - \frac{i \alpha x_{23}}{x_{12} x_{13}} e^{\alpha^2 \frac{y_{23}}{x_{23}}}, \\
        & \<O_1(x_1,y_1) V_{\alpha}(x_2,y_2) V_{-\alpha}(x_3,y_3)\>^{(0)} = - i \alpha \left( \frac{y_{12}}{x_{12}^2} - \frac{y_{13}}{x_{13}^2} \right) e^{\alpha^2 \frac{y_{23}}{x_{23}}} \\
        & \left\< \prod^n_{k=1} V_{\alpha_k}(x_k,y_k) \right\>^{(0)} = \text{exp} \left[ - \sum_{i<j} \alpha_i \alpha_j \frac{y_{ij}}{x_{ij}} \right] \delta_{0, \sum_i \alpha_i}
    \end{split}
\end{equation}

\subsection{$T \overline{T}$ deformation}

\subsubsection{Deformed Lagrangian}

As the discussion in section \ref{sec:1stTTbarJTaroot}, the $T \overline{T}$ deformed physical quantities like Lagrangian, stress tensor and action can be expanded in Taylor series by power of deformation parameter $\lambda$, and use the recursion relation \eqref{eq:recursion} to compute the corrections to the Lagrangian order by order. For example, the first several orders are
\begin{equation} \label{eq:Ln}
    \begin{split}
        & \mathcal{L}^{(0)} = (\p_y \phi)^2, \quad \mathcal{L}^{(1)} = - (\p_y \phi)^4, \quad \mathcal{L}^{(2)} = 4 (\p_y \phi)^6, \quad \mathcal{L}^{(3)} = -30 (\p_y \phi)^8, \\
        & \mathcal{L}^{(4)} = 336 (\p_y \phi)^{10}, \quad \mathcal{L}^{(5)} = - 5040 (\p_y \phi)^{12}, \quad \mathcal{L}^{(6)} = 95040 (\p_y \phi)^{14}, \\
        & \mathcal{L}^{(7)} = - 2162160 (\p_y \phi)^{16}, \quad \mathcal{L}^{(8)} = 57657600 (\p_y \phi)^{18}, \\
        & \mathcal{L}^{(9)} = -1764322560 (\p_y \phi)^{20}, \quad \mathcal{L}^{(10)} = 60949324800 (\p_y \phi)^{22},
    \end{split}
\end{equation}
where $\mathcal{L}^{(n)}$ is the $n$-th order of the expansion of the deformed Lagrangian, namely $\mathcal{L}^{[\lambda]} = \sum_{n = 0}^{\infty} \frac{\lambda^n}{n!} \mathcal{L}^{(n)}$. We observe that these terms can be cast into a general form
\begin{equation}
    \mathcal{L}^{(n)} = \frac{(-4)^n}{1 + n} \frac{\Gamma (\frac{1}{2} + n)}{\Gamma (\frac{1}{2})} (\p_y \phi)^{2n + 2},
\end{equation}
where the Gamma-function $\Gamma (x) = \int^{\infty}_0 \td t t^{x - 1} e^{-t}$ has been introduced. Implementing the summation \eqref{eq:Taylor} gives the closed form of the deformed Lagrangian
\begin{equation} \label{eq:NambuGoto}
    \mathcal{L}^{[\lambda]} = \sum_{n = 0}^{\infty} \frac{\lambda^n}{n!} \mathcal{L}^{(n)} = \frac{\sqrt{4 \lambda (\p_y \phi)^2 +1}-1}{2 \lambda}.
\end{equation}
In Appendix \ref{sec:DeformUR}, we have provided a detailed demonstration that this Lagrangian is equivalent to implementing the UR limit from the $T \overline{T}$ deformed Lagrangian of the free scalar model in the 2D relativistic CFT, as initially discussed in \cite{Cavaglia:2016oda}. For further information, refer to the review in \cite{Jiang:2019epa}. Consequently, \eqref{eq:NambuGoto} can be interpreted as the non-relativistic Nambu-Goto Lagrangian, suggesting that the $T \overline{T}$ deformation maps the local BMSFT scalar model into a non-local and non-relativistic bosonic string.

\subsubsection{Deformed correlators accurate to second order}

Now we can compute the deformed correlator constructed by the primary operators from the original data. From \eqref{eq:1stTTbar} and \eqref{eq:2ndTTbar}, to compute the deformed correlator accurate to second order of $\lambda$, one needs the data of $S^{(1)}$ and $S^{(2)}$, which we have been derived in \eqref{eq:Ln} as
\begin{equation}
    S^{(1)} = - \int \td x \td y (\p_y \phi)^4, \quad S^{(2)} = 4 \int \td x \td y (\p_y \phi)^6.
\end{equation}
Note that $(\p_y \phi)^4$ and $(\p_y \phi)^6$ are in the classical level, which need to be rewritten in the quantum level by normal ordering. However, the ways of quantization lead to different results. For example, $(\p_y \phi)^2 $ can be quantized as $ :\p_y \phi::\p_y \phi:$ or $:\p_y \phi \p_y \phi:$, which should be distinguished while calculating the deformed correlation functions, since they have distinct OPEs with other operators. Actually, the quantum version of deformed action depends on the data in the seed theory, or more precisely, it only depends on the un-deformed stress tensor $T = 2 : \p_y \phi \p_x \phi :, \ M = : \p_y \phi \p_y \phi :$ rather than $J = - 2: \p_y \phi :$, since the deformed action is only triggered by stress tensor. In particular, the first-order correction only depends on $M$ and must be $\int \td x \td y MM$, as discussed in previous sections. Moreover, the deformed action in the quantum level must be independent with $T$, since the deformed Lagrangian \eqref{eq:NambuGoto} is independent with $\p_x \phi$. In conclusion, the quantum level of the corrections are
\begin{equation}
    S^{(1)} = - \int \td x \td y MM, \quad S^{(2)} = 4 \int \td x \td y MMM
\end{equation}
Then, the deformed correlators can be derived as
\begin{equation}
    \< X_n \>_{[\lambda]}^{T \overline{T}} = \<X_n\> + \lambda \<X_n\>^{(1)}_{T \overline{T}} + \frac{\lambda^2}{2!} \<X_n\>^{(2)}_{T \overline{T}} + \mathcal{O}(\lambda^3) ,
\end{equation}
where
\begin{equation} \label{eq:XnTTbar12}
    \begin{split}
        \<X_n\>^{(1)}_{T \overline{T}} = & \int \td x \td y \<(MM) (x,y) X_n\>, \\
        \<X_n\>^{(2)}_{T \overline{T}} = & \int \td x \td y \int \td x' \td y' \<(MM) (x,y) (MM) (x',y') X_n\> \\
        & - 4 \int \td x \td y \<(MMM) (x,y) X_n\>
    \end{split}
\end{equation}
With the Ward identities \eqref{eq:WardM}, the generic form of the second order of the $T \overline{T}$ deformed BMS free scalar correlators can be derived as
\begin{equation}
    \begin{split}
        \<X_n\>^{(2)}_{T \overline{T}} = & - 4 \pi^2 \left[ \sum_{i \neq j} \frac{y_{ij}}{x_{ij}} \left( \frac{2}{x_{ij}^2} \boldsymbol{\xi}_i \boldsymbol{\xi}_j + \frac{1}{x_{ij}} \left( \boldsymbol{\xi}_i \p_{y_j} + \boldsymbol{\xi}_j \p_{y_i} \right) + \p_{y_j} \p_{y_i}  \right) \right]^2 \<X_n\> \\
        & - 8 \pi i \sum_{i \neq j \neq k} \Bigg[ \left( y_{jk} \frac{x_{jk} + x_{ji}}{x_{jk}^3} - y_{ik} \frac{x_{ik} + x_{ij}}{x_{ik}^3} \right) \frac{2 \boldsymbol{\xi}_i \boldsymbol{\xi}_j \boldsymbol{\xi}_k}{x_{ij}^3} \\
        & - \left( \frac{y_{jk}}{x_{jk}^3} + y_{ij} \frac{2 x_{ik} + x_{ij}}{x_{ij}^3} \right) \frac{3 \boldsymbol{\xi}_i \boldsymbol{\xi}_j \p_{y_k}}{x_{ik}} \\
        & + \left( \frac{y_{ik}}{x_{ik}^2} - \frac{y_{jk}}{x_{jk}^2} \right) \frac{3 \boldsymbol{\xi}_i \p_{y_j} \p_{y_k}}{x_{ij}} + \left( \frac{y_{ik}}{x_{ik}} - \frac{y_{jk}}{x_{jk}} \right) \frac{\p_{y_i} \p_{y_j} \p_{y_k}}{x_{ij}} \Bigg] \<X_n\> \\
        & + 24 \pi i \sum_{i \neq j} \frac{y_{ij}}{x_{ij}^2} \left[ 4 \frac{\boldsymbol{\xi}_i^2 \boldsymbol{\xi}_j}{x_{ij}^3} - 9 \frac{\boldsymbol{\xi}_i \boldsymbol{\xi}_j \p_{y_i}}{x_{ij}^4} - 3 \frac{\boldsymbol{\xi}_i \p_{y_i} \p_{y_j}}{x_{ij}} + \p_{y_i}^2 \p_{y_j} \right] \<X_n\>
    \end{split}
\end{equation}
The first-order $\<X_n\>^{(1)}_{T \overline{T}}$ has been derived in \eqref{eq:TTbar1st}. Actually, $S^{(n)}$ with different integers $n$ are only related to $M$, and independent with $T$. Therefore, in principle, the scheme of the integral \eqref{eq:ointall=0} is enough to derive the correction in all order. Since our purpose is just to see the flow effect to the poles of the correlators, we only need to be accurate to the second-order correction, instead of all-order correction. The rest of this sub-subsection is to manifest all the effects on the poles, by presenting some concrete examples.

\paragraph{Correlators full of vertex operators} The $n$-point vertex function in the seed theory is
\begin{equation} \label{eq:<Vertex>}
    \left\< \prod^n_{k=1} V_{\alpha_k}(x_k,y_k) \right\>^{(0)} = \text{exp} \left[ - \sum_{i<j} \alpha_i \alpha_j \frac{y_{ij}}{x_{ij}} \right] \delta_{0, \sum_i \alpha_i}.
\end{equation}
Then the first-order correction is
\begin{equation}
    \left\< \prod^n_{k=1} V_{\alpha_k} \right\>^{(1)}_{T \overline{T}} = \pi i \sum_{m \neq k} \alpha_m \alpha_k \frac{y_{km}}{x_{km}} \left[ 2 A_k A_m + 2 \frac{A_k \alpha_m}{x_{km}} - \frac{\alpha_k \alpha_m}{x_{km}^2} \right] \left\< \prod^n_{s=1} V_{\alpha_s} \right\>^{(0)}
\end{equation}
The second order correction is
\begin{equation}
    \begin{split}
        & \left\< \prod^n_{k=1} V_{\alpha_k} \right\>^{(2)}_{T \overline{T}} = \left\< \prod^n_{k=1} V_{\alpha_k} \right\>^{(0)} \times \bigg\{ - \pi^2 \left[ \sum_{m \neq k} \alpha_m \alpha_k \frac{y_{km}}{x_{km}} \left( 2 A_k A_m + 2 \frac{A_k \alpha_m}{x_{km}} - \frac{\alpha_k \alpha_m}{x_{km}^2} \right) \right]^2 \\
        & - 8 \pi i \sum_{\substack{m,k,s =1 \\ (k \neq m, k \neq s, m \neq s)}}^n \alpha_m \alpha_k \alpha_s \bigg[ \frac{A_m A_k A_s}{x_{ms}} \left( \frac{y_{mk}}{x_{mk}} - \frac{y_{sk}}{x_{sk}} \right) + \frac{3 \alpha_k A_m A_s}{2 x_{sm}} \left( \frac{y_{km}}{x_{km}^2} + \frac{y_{sk}}{x_{sk}^2} \right) \\
        & - \frac{\alpha_k \alpha_s \alpha_m}{4 x_{ms}^3} \left( y_{mk} \frac{x_{ms} + x_{mk}}{x_{mk}^3} - y_{sk} \frac{x_{ms} + x_{ks}}{x_{ks}^3} \right) + \frac{3 \alpha_k \alpha_m A_s}{x_{ms}^2} \left( \frac{y_{sk}}{x_{sk}^2} - y_{mk} \frac{2 x_{ms} + x_{mk}}{x_{mk}^3} \right) \bigg] \\
        & - 16 \pi i \sum_{m \neq k} \alpha_k \alpha_m^2 \frac{y_{km}}{x_{km}^5} \left( A_k A_m^2 x_{km}^3 - \frac{1}{2} \alpha_k \alpha_m^2 - \frac{9}{4} \alpha_k \alpha_m A_s x_{km} - 3 A_m^2 \alpha_k^2 x_{km}^2 \right) \bigg\}
    \end{split}
\end{equation}
Therefore the first and second-order corrections of the $n$-point vertex correlators are all factorized for arbitrary $n$. Further, we can easily deduce that the $T \overline{T}$ deformed $n$-point vertex correlators are factorized at all orders.

\paragraph{2-point} In seed theory, the two-point functions are
\begin{equation} \label{eq:2ptO}
    \begin{split}
        & \<O_0(x_1) O_0(x_2)\>^{(0)} = \<O_0(x',y') V_{\alpha}(x,y)\>^{(0)} = \<O_1(x',y') V_{\alpha}(x,y)\>^{(0)} = 0, \\
        & \<O_0(x_1) O_1(x_2,y_2)\>^{(0)} = \frac{1}{x_{12}^2}, \quad \<O_1(x_1,y_1) O_1(x_2,y_2)\>^{(0)} = - \frac{2y_{12}}{x_{12}^3}.
    \end{split}
\end{equation}
It is easy to confirm that the $MM$-insertion into the 2-point function of multiple primary operators in this free scalar case are all zero. Note that each order of $T\overline{T}$ deformed correlator is corrected by inserting one or more composite operators $MM$. Therefore, the correction of any order is exactly zero for a 2-point case
\begin{equation}
    \<O_a O_b\>^{(k)}_{T \overline{T}} = \<O_a V_{\alpha}\>^{(k)}_{T \overline{T}} = 0, \quad \forall k = 1,2, \cdots,
\end{equation}
which indicates that these two-point functions will not flow. In a non-perturbative way, we have
\begin{equation}
    \<O_a O_b\>_{[\lambda]}^{T \overline{T}} = \<O_a O_b\>^{(0)}, \quad  \<O_a V_{\alpha}\>_{[\lambda]}^{T \overline{T}} = \<O_a V_{\alpha}\>^{(0)}.
\end{equation}

\paragraph{3-point} The non-zero three-point functions in the seed theory are
\begin{equation} \label{eq:<OVV>}
    \begin{split}
        & \<O_0(x_1) V_{\alpha}(x_2,y_2) V_{-\alpha}(x_3,y_3)\>^{(0)} = - \frac{i \alpha x_{23}}{x_{12} x_{13}} e^{\alpha^2 \frac{y_{23}}{x_{23}}}, \\
        & \<O_1(x_1,y_1) V_{\alpha}(x_2,y_2) V_{-\alpha}(x_3,y_3)\>^{(0)} = - i \alpha \left( \frac{y_{12}}{x_{12}^2} - \frac{y_{13}}{x_{13}^2} \right) e^{\alpha^2 \frac{y_{23}}{x_{23}}}.
    \end{split}
\end{equation}
Their first-order corrections are
\begin{equation}
    \<O_{0} (\Vec{x}_1) V_{\alpha}(\Vec{x}_2) V_{-\alpha}(\Vec{x}_3)\>^{(1)}_{T \overline{T}} = 20 \pi i \alpha^4 \frac{y_{23}}{x_{23}^3} \<O_{0} (\Vec{x}_1) V_{\alpha}(\Vec{x}_2) V_{-\alpha}(\Vec{x}_3)\>^{(0)},
\end{equation}
and
\begin{equation}
    \begin{split}
        & \<O_1(x_1,y_1) V_{\alpha}(x_2,y_2) V_{-\alpha}(x_3,y_3)\>^{(1)}_{T \overline{T}} \\
        & = 2 \pi i \alpha^3 x_{23}^2 e^{\alpha^2 \frac{y_{23}}{x_{23}}} \bigg[ \frac{2x_{23}}{x_{13} x_{12}} \left( \frac{y_{13}}{x_{12}^3 x_{13}^2} + \frac{y_{23}}{x_{21}^3 x_{23}^2} \right) + \frac{y_{13}}{x_{12}^2 x_{13}^4} - \frac{y_{12}}{x_{13}^2 x_{12}^4} \\
        & + y_{23} \frac{6 x_{12} x_{13} - x_{23}^2}{x_{13}^2 x_{12}^2 x_{23}^4} + y_{23} \frac{3 x_{13} + x_{23}}{x_{12} x_{13}^2 x_{23}^4} + 5 \alpha^2 \frac{y_{23}}{x_{23}^5} \left( \frac{y_{13}}{x_{13}^2} - \frac{y_{12}}{x_{12}^2} \right) \bigg],
    \end{split}
\end{equation}
respectively. Their second-order corrections are 
\begin{equation}
    \begin{split}
        \<O_0 V_{\alpha} V_{-\alpha}\>^{(2)}_{T \overline{T}} = - \pi \alpha^6 \frac{y_{23}}{x_{23}^5} \left( 400 \pi \alpha^2 \frac{y_{23}}{x_{23}} - 756 i \right) \<O_0 V_{\alpha} V_{-\alpha}\>^{(0)}
    \end{split}
\end{equation}
and
\begin{align}
    & \<O_1(x_1,y_1) V_{\alpha}(x_2,y_2) V_{-\alpha}(x_3,y_3)\>^{(2)}_{T \overline{T}} \notag \\
    & = 2 \pi \alpha^5 x_{23}^4 e^{\alpha^2 \frac{y_{23}}{x_{23}}} \bigg\{ \frac{12 x_{23}}{x_{12} x_{13}^6} \left[ \frac{y_{12}}{x_{12}^4} \left( \frac{x_{13}}{x_{12}} + 1 \right) + \frac{y_{23}}{x_{23}^7} \left( 5 x_{13}^3 + 5 x_{13}^2 x_{23} + 3 x_{13} x_{23}^2 + x_{23}^3 \right) \right] \notag \\
    & \quad - \frac{5}{2} \left[ \frac{y_{32}}{x_{12} x_{13} x_{23}^5} \left( \frac{70}{x_{23}^3} - \frac{5}{x_{23} x_{13} x_{12}} + \frac{x_{23}}{x_{13}^2 x_{12}^2} \left( 2 - \frac{x_{23}^2}{x_{13} x_{12}} \right) \right) + \frac{y_{12}}{x_{12} x_{13}^4} - \frac{y_{13}}{x_{13} x_{12}^4} \right] \notag \\
    & \quad - 252 \left( \frac{y_{12}}{x_{12}^2} - \frac{y_{13}}{x_{13}^2} \right) \frac{y_{23}}{x_{23}^9} \bigg\} + 40 \pi^2 i \frac{y_{23}}{x_{23}} \alpha^7 e^{\alpha^2 \frac{y_{23}}{x_{23}}} \bigg\{ \left( \frac{y_{12}}{x_{12}^2} - \frac{y_{13}}{x_{13}^2} \right) \left( \frac{1}{x_{12}^2 x_{13}^2} + \frac{5}{2} \alpha^5 \frac{y_{23}}{x_{23}^5} \right) \notag \\
    & \quad + y_{23} \frac{x_{23}^2 - 6 x_{12} x_{13}}{x_{12}^2 x_{13}^2 x_{23}^4} + \frac{2 x_{23}}{x_{12}^4 x_{13}} \left[ \frac{y_{23}}{x_{23}^3} (x_{23} + x_{21}) - \frac{y_{13}}{x_{13}^3} (x_{13} + x_{12}) \right] \bigg\}
\end{align}
respectively.

\subsection{$JT_{\mu}$ deformation}

The $JT_{\mu}$ deformation can be constructed from the affine current $j$ and the stress tensor growing in the seed theory as
\begin{equation} \label{eq:JTL}
    \frac{\p \mathcal{L}^{[\lambda]}}{\p \lambda^{\mu}} = \epsilon_{\alpha \beta} j^{\alpha}_{[\lambda]} T^{\beta [\lambda]}_{\ \mu} = j^y_{[\lambda]} T^{x [\lambda]}_{\ \mu} - j^x_{[\lambda]} T^{y [\lambda]}_{\ \mu}
\end{equation}
In a perturbative level, since the composite operator of $JT_{\mu}$ deformation is a vector operator with two components, the quantities needed to be computed should be expanded by the power of two coupling constants $\lambda^0, \lambda^1$, and they might be mixed while implementing the Taylor expansion, which will make the problem more complex. However, things will become much easier if the two coupling constants are not mixed. Fortunately, in the case we discussed here, we will prove that the contribution of $\lambda^0$ is zero, or in other words, there is only one coupling constant $\lambda^1$. We expand the quantities as follows
\begin{equation} \label{eq:expandLTj}
    \begin{split}
        \mathcal{L}^{[\lambda]} & = \sum_{n = 0}^{\infty} \frac{1}{n!} \lambda^{\mu_1} \cdots \lambda^{\mu_n} \mathcal{L}_{\mu_1 \cdots \mu_n}, \\
        T^{\alpha [\lambda]}_{\ \mu} & = \sum_{n = 0}^{\infty} \frac{1}{n!} \lambda^{\mu_1} \cdots \lambda^{\mu_n} (T^{\alpha}_{\ \mu})_{\mu_1 \cdots \mu_n}, \\
        j^{\alpha}_{[\lambda]} & = \sum_{n = 0}^{\infty} \frac{1}{n!} \lambda^{\mu_1} \cdots \lambda^{\mu_n} (j^{\alpha})_{\mu_1 \cdots \mu_n},
    \end{split}
\end{equation}
where the extra indices $\mu_1 \cdots \mu_n$ are all symmetrical, namely
\begin{equation}
    \mathcal{L}_{\mu_1 \cdots \mu_n} = \mathcal{L}_{(\mu_1 \cdots \mu_n)}, \quad (T^{\alpha}_{\ \mu})_{\mu_1 \cdots \mu_n} = (T^{\alpha}_{\ \mu})_{(\mu_1 \cdots \mu_n)}, \quad (j^{\alpha})_{\mu_1 \cdots \mu_n} = (j^{\alpha})_{(\mu_1 \cdots \mu_n)}.
\end{equation}
Then the recursion relation like \eqref{eq:recursion} can be derived as
\begin{equation}
    \begin{split}
        \mathcal{L}_{\mu_1 \cdots \mu_n \mu_{n+1}} & = \epsilon_{\alpha \beta} \sum_{i = 0}^n C_n^i (j^{\alpha})_{\mu_{i+1} \cdots \mu_n} (T^{\beta}_{\ \mu_{n+1}})_{\mu_1 \cdots \mu_i} \\
        & = \sum_{i = 0}^n C_n^i \left[ (j^y)_{\mu_{i+1} \cdots \mu_n} (T^x_{\ \mu_{n+1}})_{\mu_1 \cdots \mu_i} - (j^x)_{\mu_{i+1} \cdots \mu_n} (T^y_{\ \mu_{n+1}})_{\mu_1 \cdots \mu_i} \right] 
    \end{split}
\end{equation}
where
\begin{equation} \label{eq:Tjn-th}
    (T^{\alpha}_{\ \mu})_{\mu_1 \cdots \mu_n} = \frac{\p \mathcal{L}_{\mu_1 \cdots \mu_n}}{\p (\p_{\alpha} \phi)} \p_{\mu} \phi - \delta^{\alpha}_{\ \mu} \mathcal{L}_{\mu_1 \cdots \mu_n}, \quad (j^{\beta})_{\mu_{i+1} \cdots \mu_n} = - \frac{\p \mathcal{L}_{\mu_{i+1} \cdots \mu_n}}{\p (\p_{\beta} \phi)}.
\end{equation}
and $(T^{\alpha}_{\ \mu})_{\mu_1 \cdots \mu_i} = T^{\alpha [0]}_{\ \mu}$ for $i = 0$ while $(j^{\beta})_{\mu_{i+1} \cdots \mu_n} = j^{\beta}_{[0]}$ for $i = n$. We then prove that
\begin{equation} \label{eq:jx=0}
    (j^x)_{\mu_1 \cdots \mu_n} = - \frac{\p \mathcal{L}_{\mu_1 \cdots \mu_n}}{\p (\p_x \phi)} = 0, 
    \quad \forall n \in \mathbb{N}
\end{equation}
by using induction on $n$. For $n = 0$, then
\begin{equation}
    \begin{split}
        & \mathcal{L}_{\mu_1 \cdots \mu_n} \big|_{n = 0} = \mathcal{L}_{[0]}, \quad j^x_{\mu_1 \cdots \mu_n} \big|_{n = 0} = j^x_{[0]} = - \frac{\p \mathcal{L}_{[0]}}{\p (\p_x \phi)} = 0, \\
    \end{split}
\end{equation}
which clearly satisfies the $n = 0$ case in \eqref{eq:jx=0}. Suppose \eqref{eq:jx=0} is true for $n$ less than $m+1$
\begin{equation} \label{eq:jx=0,m}
    (j^x)_{\mu_1 \cdots \mu_n} = - \frac{\p \mathcal{L}_{\mu_1 \cdots \mu_n}}{\p (\p_x \phi)} = 0, \quad n = 0, 1, 2, \cdots m,
\end{equation}
then the $m+1$-th order of Lagrangian is
\begin{equation}
    \mathcal{L}_{\mu_1 \cdots \mu_m \mu_{m+1}} = \sum_{i = 0}^m C_n^i (j^y)_{\mu_{i+1} \cdots \mu_m} (T^x_{\ \mu_{m+1}})_{\mu_1 \cdots \mu_i}
\end{equation}
From the definition \eqref{eq:Tjn-th}, together with the induction hypothesis \eqref{eq:jx=0,m}, one can easily express the currents lower than $m+1$-th order as
\begin{equation}
    (j^y)_{\mu_{i+1} \cdots \mu_m} = - \frac{\p \mathcal{L}_{\mu_{i+1} \cdots \mu_m}}{\p (\p_y \phi)},
\end{equation}
\begin{equation}
    (T^x_{\ \mu_{m+1}})_{\mu_1 \cdots \mu_i} = \frac{\p \mathcal{L}_{\mu_1 \cdots \mu_i}}{\p (\p_x \phi)} \p_{\mu_{m+1}} \phi - \delta^x_{\ \mu_{m+1}} \mathcal{L}_{\mu_1 \cdots \mu_i} = - \delta^x_{\ \mu_{m+1}} \mathcal{L}_{\mu_1 \cdots \mu_i},
\end{equation}
which indicate that they are all independent with $\p_x \phi$ because of \eqref{eq:jx=0,m}
\begin{align}
    & \frac{\p (j^y)_{\mu_{i+1} \cdots \mu_m}}{\p (\p_x \phi)} = - \frac{\p^2 \mathcal{L}_{\mu_{i+1} \cdots \mu_m}}{\p (\p_x \phi) \p (\p_y \phi)} = 0, \\
    & \frac{\p (T^x_{\ \mu_{m+1}})_{\mu_1 \cdots \mu_i}}{\p (\p_x \phi)} = - \delta^x_{\ \mu_{m+1}} \frac{\p \mathcal{L}_{\mu_1 \cdots \mu_i}}{\p (\p_x \phi)} = 0.
\end{align}
Thus we deduce that the current in the $m+1$-th order correction as
\begin{equation}
    (j^x)_{\mu_1 \cdots \mu_{m+1}} = - \frac{\p \mathcal{L}_{\mu_1 \cdots \mu_{m+1}}}{\p (\p_x \phi)} = 0
\end{equation}
Thus the eq \eqref{eq:jx=0} has been proved, in other words, only $j^y_{[\lambda]}$ contribute to the corrections of deformation. It turns out that the $n$-th order correction of the deformed Lagrangian is
\begin{equation}
    \mathcal{L}_{\mu_1 \cdots \mu_n} = \sum_{i = 0}^n C_n^i (j^y)_{\mu_{i+1} \cdots \mu_n} (T^x_{\ \mu_{n+1}})_{\mu_1 \cdots \mu_i}, \quad \forall n \in \mathbb{N},
\end{equation}
By using the definition of the currents \eqref{eq:Tjn-th}, one obtains
\begin{equation}
    \begin{split}
        & (T^x_{\ x})_{\mu_1 \cdots \mu_n} = - \mathcal{L}_{\mu_1 \cdots \mu_n}, \quad (T^x_{\ y})_{\mu_1 \cdots \mu_n} = 0, \\
        & (j^y)_{\mu_{i+1} \cdots \mu_n} = - \frac{\p \mathcal{L}_{\mu_1 \cdots \mu_n}}{\p (\p_y \phi)}, \quad (j^x)_{\mu_{i+1} \cdots \mu_n} = 0
    \end{split}
\end{equation}
which indicates that
\begin{equation} \label{eq:jxTxy=0}
    T^{x [\lambda]}_{\ y} = 0, \quad j^x_{[\lambda]} = 0,
\end{equation}
or in other words, the only choice of the vector index $\mu$ in the definition \eqref{eq:JTL} must be $\mu = x$, otherwise, the RHS of \eqref{eq:JTL} will vanish
\begin{equation}
    \frac{\p \mathcal{L}^{[\lambda]}}{\p \lambda^1} = j^y_{[\lambda]} T^{x [\lambda]}_{\ x}, \quad \frac{\p \mathcal{L}^{[\lambda]}}{\p \lambda^0} = j^y_{[\lambda]} T^{x [\lambda]}_{\ y} = 0.
\end{equation}
So the coupling constant $\lambda^0$ will not appear in the definition \eqref{eq:JTL}. Hence, the corrections of currents will no longer depend on $\lambda^0$ either since they are all derived from the Lagrangian.

\subsubsection{Deformed Lagrangian}

Since all of the quantities in the deformed story only depend on $\lambda^1$, then \eqref{eq:expandLTj} will no longer be the double-coefficient expansion. For convenience, we substitute $\lambda$ for $\lambda^1$, and use \eqref{eq:jxTxy=0} to rewrite the definition \eqref{eq:JTL} as
\begin{equation} \label{eq:1lambdaJTaL}
    \frac{\p \mathcal{L}^{[\lambda]}}{\p \lambda} = j^y_{[\lambda]} T^{x [\lambda]}_{\ x}.
\end{equation}
Then the expansion \eqref{eq:expandLTj} should be rewritten as
\begin{equation} \label{eq:TaylorJTmu}
    \mathcal{L}^{[\lambda]} = \sum_{n = 0}^{\infty} \frac{\lambda^n}{n!} \mathcal{L}^{(n)}, \quad T^{x [\lambda]}_{\ x} = \sum_{n = 0}^{\infty} \frac{\lambda^n}{n!} T^{x (n)}_{\ x}, \quad j^y_{[\lambda]} = \sum_{n = 0}^{\infty} \frac{\lambda^n}{n!} j^y_{(0)},
\end{equation}
like we did in $T \overline{T}$ case. Then the recursion relation can be simply rewritten as
\begin{equation} \label{eq:recursionJTa}
    \mathcal{L}^{(n+1)} = \sum_{i = 0}^n C_n^i j^y_{(n-i)} T^{x (i)}_{\ x}, \quad T^{x (i)}_{\ x} = - \mathcal{L}^{(i)}, \quad j^y_{(n-i)} = - \frac{\p \mathcal{L}^{(n - i)}}{\p (\p_y \phi)}.
\end{equation}
Then the deformed Lagrangian can be derived order by order from the above recursion relation. For example, we show the result for the first several orders of the deformed Lagrangian
\begin{equation} \label{eq:L10JTa}
    \begin{split}
        & \mathcal{L}^{(0)} = (\p_y \phi)^2, \quad \mathcal{L}^{(1)} = 2 (\p_y \phi)^3, \quad \mathcal{L}^{(2)} = 10 (\p_y \phi)^4, \quad \mathcal{L}^{(3)} = 84 (\p_y \phi)^5, \\
        & \mathcal{L}^{(4)} = 1008 (\p_y \phi)^6, \quad \mathcal{L}^{(5)} = 15840 (\p_y \phi)^7, \quad \mathcal{L}^{(6)} = 308880 (\p_y \phi)^8, \\
        & \mathcal{L}^{(7)} = 7207200 (\p_y \phi)^9, \quad \mathcal{L}^{(8)} = 196035840 (\p_y \phi)^{10}, \\
        & \mathcal{L}^{(9)} = 6094932480 (\p_y \phi)^{11}, \quad \mathcal{L}^{(10)} = 213322636800 (\p_y \phi)^{12}.
    \end{split}
\end{equation}
These terms can also be cast into a general form,
\begin{equation}
    \mathcal{L}^{(n)} = \frac{2^{1 + 2n}}{(1 + n) (2 + n)} \frac{\Gamma (\frac{3}{2} + n)}{ \Gamma (\frac{3}{2})} (\p_y \phi)^{n+2},
\end{equation}
where $\Gamma (x)$ is Gamma-function, which has been introduced while discussing $T \Bar{T}$ deformation. Then, implementing the summation \eqref{eq:TaylorJTmu} gives the closed form of the $JT_{\mu}$ deformed Lagrangian
\begin{equation}
    \mathcal{L}^{[\lambda]} = \sum_{n = 0}^{\infty} \lambda^n \mathcal{L}^{(n)} = \frac{1 - 2 \lambda \p_y \phi - \sqrt{1 - 4 \lambda \p_y \phi }}{2 \lambda ^2}.
\end{equation}
Similarly to \eqref{eq:NambuGoto}, this Lagrangian also indicates that the $JT_{\mu}$ deformation maps the local BMS free scalar to a non-local theory, aligning with the well-known characteristic of an irrelevant deformation.

\subsubsection{Deformed correlation functions}

By expanding the action as $S^{[\lambda]} = \sum_n \frac{\lambda^n}{n!} S^{(n)}$, where $S^{(n)} = \int \td x \td y \mathcal{L}^{(n)}$, one can similarly derive the corrections for the deformed correlation function as
\begin{equation}
    \<X_n\>_{[\lambda]}^{JT_{\mu}} = \sum_{n = 0}^{\infty} \frac{\lambda^n}{n!} \<X_n\>^{(n)},
\end{equation}
where the corrections are formally the same as \eqref{eq:0thTTbar} \eqref{eq:1stTTbar} and \eqref{eq:2ndTTbar}
\begin{align}
    \<X_n\>^{(0)} = & \<X_n\>_{[\lambda = 0]}^{JT_{\mu}} = \<X_n\> \label{eq:0thJTa} \\
    \<X_n\>^{(1)}_{JT_{\mu}} = & \left\< S^{(1)} \right\> \<X_n\> - \left\< S^{(1)} X_n \right\>, \label{eq:1stJTa} \\
    \<X_n\>^{(2)}_{JT_{\mu}} = & \left\< S^{(1)} S^{(1)} X_n \right\> - \left\< S^{(1)} S^{(1)} \right\> \< X_n \> + \left\< S^{(2)} \right\> \<X_n\> -  \left\< S^{(2)} X_n \right\> \notag \\
    & + 2 \left\< S^{(1)} \right\>^2 \<X_n\> - 2 \left\< S^{(1)} \right\> \left\< S^{(1)} X_n \right\> \label{eq:2ndJTa} \\
    \vdots & \notag
\end{align}
The first two corrections of the action are \eqref{eq:L10JTa}
\begin{equation}
    S^{(1)} = 2 \int \td x \td y (\p_y \phi)^3 
    , \quad S^{(2)} = 10 \int \td x \td y (\p_y \phi)^4 
\end{equation}
%
which are the corrections of action in the classical level. Similar to the discussion of $T \overline{T}$, to derive the corrected correlation functions perturbatively, the above corrections of action should be promoted to quantum level by using normal ordering. However, things will become more complex in $JT_{\mu}$ case than that in $T \overline{T}$ case, because the $JT_{\mu}$ deformation is triggered by $U(1)$ current and stress tensor, which indicates that the deformed action in quantum level is not only dependent on the stress tensor $T,M$ in the seed theory, but it also depends on the un-deformed $U(1)$ current $J_x (= J)$. We have proved that the $JT_{\mu}$ deformed Lagrangian of the free scalar is independent with $\p_x \phi$, such that the deformed quantities are not dependent on $T$. The first-order correction here must match the most generic case discussed in the subsection \ref{sec:JTa1st} with $J_y = 0$ here, namely the first-order correction is always quantized as
\begin{equation}
    S^{(1)} = \int \td x \td y \mathcal{L}^{(1)}, \quad \mathcal{L}^{(1)} = - JM
\end{equation}
Then the first-order corrected $JT_{\mu}$ deformed correlator is
\begin{equation}
    \<X_n\>^{(1)}_{JT_{\mu}} = \int \td x \td y \<JM X_n\>
\end{equation}
Unlike $T \overline{T}$ case, $(\p_y \phi)^4$ can be quantized as $JJM$ or $MM$ in $JT_{\mu}$ deformation. So the quantization of the second order correction can be expressed as the linear combination of $\int \td x \td y JJM$ and $\int \td x \td y MM$. Fortunately, the coefficient in front of them can be uniquely fixed. We can see this precisely from the recursion relation \eqref{eq:recursionJTa}
\begin{equation}
    \mathcal{L}^{(2)} = j^y_{(0)} T^{x(1)}_{\ x} + j^y_{(1)} T^{x(0)}_{\ x} = JJM - M \left( \frac{\p J}{\p (\p_y \phi)} M + J \frac{\p M}{\p (\p_y \phi)} \right),
\end{equation}
where we used
\begin{equation*}
    T^{x (1)}_{\ x} = - \mathcal{L}^{(1)} = JM, \quad j^y_{(0)} = J, \quad T^{x (0)}_{\ x} = -M, \quad j^y_{(1)} = \frac{\p J}{\p (\p_y \phi)} M + J \frac{\p M}{\p (\p_y \phi)}.
\end{equation*}
As we discussed before, $\frac{\p J}{\p (\p_y \phi)}, \frac{\p M}{\p (\p_y \phi)}$ only depend on $J,M$, so the quantization of $\frac{\p J}{\p (\p_y \phi)}, \frac{\p M}{\p (\p_y \phi)}$ are unique
\begin{equation}
    \frac{\p J}{\p (\p_y \phi)} = -2, \quad \frac{\p M}{\p (\p_y \phi)} = 2 : \p_y \phi : = - J
\end{equation}
Therefore
\begin{equation}
    S^{(2)} = \int \td x \td y \mathcal{L}^{(2)}, \quad \mathcal{L}^{(2)} = 2 (JJM + MM).
\end{equation}
Then the second order corrected $JT_{\mu}$ deformed correlator is
\begin{equation}
    \begin{split}
        \<X_n\>^{(2)}_{JT_{\mu}} = & \int \td x \td y \int \td x' \td y' \< JM(x,y) JM(x',y') X_n\> \\
        & - 2 \int \td x \td y \<JJM (x,y) X_n\> - 2 \<X_n\>^{(1)}_{T \overline{T}}
    \end{split}
\end{equation}
Then we can use the Ward identities \eqref{eq:WardM} and \eqref{eq:JWard} $\mathcal{F}_i = 0$ 
to compute the generic form of $\<X_n\>^{(1)}_{JT_{\mu}}$ and $\<X_n\>^{(2)}_{JT_{\mu}}$ in free scalar case as
\begin{equation}
    \<X_n\>^{(1)}_{JT_{\mu}} = - 2 \pi \sum_{i \neq j} \frac{y_{ij}}{x_{ij}^2} \mathcal{G}_i (\boldsymbol{\xi}_j + x_{ij} \p_{y_j}) \<X_n\>
\end{equation}
and
\begin{equation}
    \begin{split}
        \<X_n\>^{(2)}_{JT_{\mu}} & = - 2 \<X_n\>_{T \overline{T}}^{(1)} + 4 \pi^2 \left[ \sum_{i \neq j} \frac{y_{ij}}{x_{ij}^2} \mathcal{G}_i (\boldsymbol{\xi}_j + x_{ij} \p_{y_j}) \right]^2 \<X_n\> \\
        & + 4 \pi i \bigg\{ \sum_{i \neq j \neq k} \frac{\mathcal{G}_i \mathcal{G}_j}{x_{ij}} \left[ \frac{y_{ik}}{x_{ik}^2} (\boldsymbol{\xi}_k + x_{ik} \p_{y_k}) - \frac{y_{jk}}{x_{jk}^2} (\boldsymbol{\xi}_k + x_{jk} \p_{y_k}) \right] \\
        & + 2 \sum_{i \neq j} \mathcal{G}_i \mathcal{G}_j \frac{y_{ij}}{x_{ij}^3} (\boldsymbol{\xi}_j + x_{ij} \p_{y_j}) - \sum_{i \neq j} \mathcal{G}_i^2 \frac{y_{ij}}{x_{ij}^3} (\boldsymbol{\xi}_j + 2 x_{ij} \p_{y_j}) \bigg\} \<X_n\>
    \end{split}
\end{equation}
where $\<X_n\>_{T \overline{T}}^{(1)}$ has been derived in \eqref{eq:TTbar1st}.

\subsubsection{Examples}

Then we need to manifest the deformed poles by showing some examples, like we did in $T \overline{T}$ deformation. We will compute the deformed correlators whose seeds have been presented in \eqref{eq:Xscalar0}.

\paragraph{Correlators full of vertex operators} The first-order correction consisting of $n$-point vertex operators is
\begin{equation}
    \<\prod^n_{k = 1} V_{\alpha_k} (x_k, y_k) \>^{(1)}_{JT_{\mu}} = \pi i \sum_{i \neq j} \frac{y_{ij}}{x_{ij}^2} \alpha_i \alpha_j \bigg(-\alpha_j + x_{ij} \sum_{s (\neq j)} \frac{\alpha_s}{x_{sj}} \bigg) \<\prod^n_{k = 1} V_{\alpha_k} (x_k, y_k) \>^{(0)}
\end{equation}
The second order correction is
\begin{equation}
    \begin{split}
        & \<\prod^n_{p = 1} V_{\alpha_p} (x_p, y_p) \>^{(2)}_{JT_{\mu}} \\
        & = - 2 \<X_n\>_{T \overline{T}}^{(1)} - \pi^2 \bigg[ \sum_{i \neq j} \frac{y_{ij}}{x_{ij}^2} \alpha_i \alpha_j \bigg( \alpha_j - x_{ij} \sum_{s (\neq j)} \frac{\alpha_s}{x_{sj}} \bigg) \bigg]^2 \<X_n\> \\
        & + 2 \pi i \bigg\{ \sum_{i \neq j \neq k} \frac{\alpha_i \alpha_j \alpha_k}{x_{ij}} \bigg[ \frac{y_{ik}}{x_{ik}^2} \bigg( \alpha_k - x_{ik} \sum_{s (\neq k)} \frac{\alpha_s}{x_{sk}} \bigg) - \frac{y_{jk}}{x_{jk}^2} \bigg( \alpha_k - x_{jk} \sum_{s (\neq k)} \frac{\alpha_s}{x_{sk}} \bigg) \bigg] \\
        & + \sum_{i \neq j} \alpha_i \alpha_j \frac{y_{ij}}{x_{ij}^3} \bigg[ \alpha_j^2 + \alpha_i \alpha_j - x_{ij} (2 \alpha_i + \alpha_j) \sum_{s (\neq j)} \frac{\alpha_s}{x_{sj}} \bigg] \bigg\} \<\prod^n_{p = 1} V_{\alpha_p} (x_p, y_p) \>
    \end{split}
\end{equation}
%
The first and second-order corrections are all factorized.

\paragraph{2-point} Note that the two-point functions in \eqref{eq:Xscalar0} vanishes while inserting the operators $\mathcal{G}_i$ defined in \eqref{eq:G}. Therefore, the first-order correction for two-point functions all vanished, and the second-order corrections $\<X_2\>^{(2)}_{JT_{\mu}} = -2 \<X_2\>^{(1)}_{T \overline{T}}$, which are the first-order corrections of $T \overline{T}$ deformation derived in \eqref{eq:2ptO}. Moreover, \eqref{eq:2ptO} shows that $\<X_2\>^{(1)}_{T \overline{T}}$ are all vanished. Therefore, the first-order and second-order correction of the two-point functions all vanished, namely
\begin{equation}
    \<X_2\>_{JT_{\mu}}^{(1)} = 0, \quad \<X_2\>_{JT_{\mu}}^{(2)} = 0.
\end{equation}

\paragraph{3-point} The first-order corrections of the three-point functions presented in \eqref{eq:Xscalar0} are
\begin{equation}
    \<O_0(\Vec{x}_1) V_{\alpha}(\Vec{x}_2) V_{-\alpha}(\Vec{x}_3)\>^{(1)}_{JT_{\mu}} = - 6 \pi i \alpha^3 \frac{y_{23}}{x_{23}^2} \<O_0(\Vec{x}_1) V_{\alpha}(\Vec{x}_2) V_{-\alpha}(\Vec{x}_3)\>
\end{equation}
and
\begin{align}
    & \<O_1(\Vec{x}_1) V_{\alpha}(\Vec{x}_2) V_{-\alpha}(\Vec{x}_3)\>^{(1)}_{JT_{\mu}} \notag \\
    & = 2 \pi \alpha^2 \left[ \frac{y_{12}}{x_{12}^2} \left( 2 \frac{x_{12}}{x_{23}} - 1 \right) - \frac{y_{13}}{x_{13}^2} \left( 2 \frac{x_{13}}{x_{23}} + 1 \right) \right] \< V_{\alpha}(\Vec{x}_2) V_{- \alpha} (\Vec{x}_3) \>^{(0)} \notag \\
    & + \frac{2 \pi i \alpha}{x_{12} x_{13} x_{23}} \left[ \left( \frac{y_{23}}{x_{23}} + \frac{y_{12}}{x_{12}} - \frac{y_{13}}{x_{13}} \right) x_{12}^2 - \left( \frac{y_{23}}{x_{23}} + \frac{y_{13}}{x_{13}} - \frac{y_{12}}{x_{12}} \right) x_{13}^2 \right] \<O_0 (\Vec{x}_1) V_{\alpha} (\Vec{x}_2) V_{- \alpha} (\Vec{x}_3) \>^{(0)} \notag \\
    & + 2 \pi i \alpha \left( \frac{x_{23}}{x_{12} x_{13}} - 3 \alpha^2 \frac{y_{23}}{x_{23}^2} \right) \<O_1 (\Vec{x}_1) V_{\alpha} (\Vec{x}_2) V_{- \alpha} (\Vec{x}_3) \>^{(0)}
\end{align}
The second-order corrections are
\begin{equation}
    \begin{split}
         & \< O_0 V_{\alpha} V_{-\alpha} \>^{(2)}_{JT_{\mu}} = - 12 \pi \alpha^4 \frac{y_{23}}{x_{23}^3} \left( 7 i + 3 \pi \alpha^2 \frac{y_{23}}{x_{23}} \right) \< O_0 V_{\alpha} V_{-\alpha} \>^{(0)}
    \end{split}
\end{equation}
and
\begin{equation}
    \begin{split}
         & \<O_1(\Vec{x}_1) V_{\alpha}(\Vec{x}_2) V_{-\alpha}(\Vec{x}_3)\>^{(2)}_{JT_{\mu}} \\
         & = i \alpha^3 \bigg[ \frac{y_{23}}{2 x_{23}} \left( \frac{1}{x_{12}^2} - \frac{1}{x_{13}^2} + \frac{2}{x_{12} x_{13}} + \frac{24 \pi i}{x_{23} x_{12}} \right) + 4 \pi i \left( \frac{y_{12}}{x_{12}^3} - \frac{y_{13}}{x_{13}^3} \right) \\
         &  \quad - \frac{1}{x_{23}} \left( \frac{\alpha^2}{2} \frac{y_{23}}{x_{23}} + 8 \pi i \right) \left( \frac{y_{12}}{x_{12}^2} + \frac{y_{13}}{x_{13}^2} \right) - \frac{8 \pi i y_{13}}{x_{12} x_{13}} \left( \frac{1}{2 x_{13}} + \frac{1}{x_{23}} \right) \\
         & \quad - \left( \frac{y_{12}}{x_{12}} - \frac{y_{13}}{x_{13}} \right) \left( \frac{x_{23}^2}{2 x_{12}^2 x_{13}^2} + 5 \alpha^2 \frac{y_{23}}{x_{23}^3} + \frac{3}{x_{23}^2} \right) \bigg] \< V_{\alpha}(\Vec{x}_2) V_{-\alpha}(\Vec{x}_3) \>^{(0)} \\
         & + \alpha^2 \bigg[ \left( \frac{y_{12}}{x_{12}} - \frac{y_{13}}{x_{13}} \right) \left( \frac{1}{x_{12}^2} - \frac{1}{x_{13}^2} \right) \left( 1 + 8 \pi i \frac{x_{12} x_{13}}{x_{23}^2} \right) - \frac{y_{23}}{x_{12} x_{13} x_{23}} \\
         & \quad + \frac{1}{x_{23}} \left( 8 \pi i + 3 \alpha^2 \frac{y_{23}}{x_{23}} \right) \left( \frac{y_{12}}{x_{12}^2} - \frac{y_{13}}{x_{13}^2} \right) + 16 \pi i \frac{y_{23}}{x_{23}^3} \frac{x_{12}^2 + x_{13}^2}{x_{12} x_{13}} \\
         & \quad - 8 \pi i \left( \frac{y_{12}}{x_{12}^2} + \frac{y_{13}}{x_{13}^2} \right) \left( \frac{1}{x_{12}} + \frac{1}{x_{13}} \right) + 4 \pi i \left( \frac{y_{12}}{x_{12}^3} + \frac{y_{13}}{x_{13}^3} \right) \\
         & \quad - \frac{y_{23}}{x_{23}^3} \left( 6 \alpha^2 \frac{y_{23}}{x_{23}} + 2 \right) \left( \frac{x_{23}^2}{x_{12} x_{13}} + 2 \right) \bigg] \<O_0(\Vec{x}_1) V_{\alpha}(\Vec{x}_2) V_{-\alpha}(\Vec{x}_3)\>^{(0)} \\
         & + 3 \alpha^4 \frac{y_{23}}{x_{23}} \left[ \frac{1}{x_{23}^2} \left( 3 \alpha^2 \frac{y_{23}}{x_{23}} + 2 + 12 \pi i \right) - \frac{1}{x_{12} x_{13}} \right] \<O_1(\Vec{x}_1) V_{\alpha}(\Vec{x}_2) V_{-\alpha}(\Vec{x}_3)\>^{(0)}
    \end{split}
\end{equation}

\subsection{Root-$T \overline{T}$ deformation} \label{sec:rootscalar}

The data of the seed theory is \eqref{eq:ScalarL}
\begin{equation}
    \mathcal{L}^{(0)} = (\p_y \phi)^2 = M.
\end{equation}
As discussed before, the first-order correction to the Lagrangian is proportion to $\mathcal{L}^{(0)}$
\begin{equation}
    \mathcal{L}^{(1)} = \mathcal{L}^{(0)} = M.
\end{equation}
Then the first-order correction to the stress tensor is
\begin{equation}
    T^{A (1)}_{\ B} = \frac{\p \mathcal{L}^{(1)}}{\p (\p_A \phi)} \p_B \phi - \delta^A_{\ B} \mathcal{L}^{(1)} = T^{A (0)}_{\ B}.
\end{equation}
From \eqref{eq:rootrecursion}, one can then derive the second-order correction to the Lagrangian as
\begin{equation}
    \begin{split}
        \mathcal{L}^{(2)} & = \frac{1}{2 \sqrt{M^2}} \sum_{i = 0}^1 \left( \frac{1}{2} T^{A (i)}_{\ B} T^{B (1-i)}_{\ A} - \frac{1}{4} T^{A (i)}_{\ A} T^{B (1 - i)}_{\ B} \right) \\
        & = \frac{1}{M} \left( \frac{1}{2} T^{A (0)}_{\ B} T^{B (0)}_{\ A} - \frac{1}{4} T^{A (0)}_{\ A} T^{B (0)}_{\ B} \right) = M = \mathcal{L}^{(0)}.
    \end{split}
\end{equation}
This gives us an insight that each order of the Lagrangian is proportional to $\mathcal{L}^{(0)}$, namely
\begin{equation} \label{eq:Ln=anL0}
    \mathcal{L}^{(n)} = a_n \mathcal{L}^{(0)}, \quad T^{A (n)}_{\ B} = a_n T^{A (0)}_{\ B}, \quad a_0 = 1.
\end{equation}
where $a_n$-s are real numbers. 

This can be proved by using induction on the correction order $n$. To get start, the \eqref{eq:Ln=anL0} is right for $n = 0,1,2$ with $a_0 = a_1 = a_2 = 1$. Then, suppose \eqref{eq:Ln=anL0} is true for $\forall n = 1, 2, \cdots, n_0$. So the relation \eqref{eq:rootrecursion} can be rewritten as
\begin{equation}
    \begin{split}
        \sum_{n = 0}^{\infty} \frac{\lambda^n}{n!} \mathcal{L}^{(n+1)} = & \Bigg\{ M^2 \left[ 1 + \sum_{n = 1}^{n_0} \frac{\lambda^n}{n!} \left(\sum_{i = 1}^n a_i a_{n - i} C_{n_0}^n \right) \right] \\
        & + \sum_{n = n_0 +1}^{\infty} \frac{\lambda^n}{n!} \left[ \sum_{i = 0}^n C_n^i \left( \frac{1}{2} T^{A (i)}_{\ B} T^{B (n-i)}_{\ A} - \frac{1}{4} T^{A (i)}_{\ A} T^{B (n - i)}_{\ B} \right) \right] \Bigg\}^{\frac{1}{2}},
    \end{split}
\end{equation}
where the second line of the above equation will contribute to the higher power of $\lambda$ after the Taylor expansion, while the first line will contribute to the $\lambda^{n_0}$ power, which indicates that $\mathcal{L}^{(n_0 + 1)}$ is also proportion to $\mathcal{L}^{(0)}$. Thus $T^{A (n_0 +1)}_{\ B}$ is also proportion to $T^{A (0)}_{\ B}$ with the same coefficient as $\mathcal{L}^{(n_0 + 1)}$. Therefore the eq \eqref{eq:rootrecursion} has been proven.

Then the deformed Lagrangian and stress tensor can be rewritten as the following factorized form
\begin{equation} \label{eq:factorize}
    \mathcal{L}^{[\lambda]} = f(\lambda) \mathcal{L}^{(0)}, \quad T^{A [\lambda]}_{\ B} = f(\lambda) T^{A (0)}_{\ B},
\end{equation}
where
\begin{equation}
    f (\lambda) = \sum_{n = 0}^{\infty} \frac{\lambda^n}{n!} a_n.
\end{equation}
Plugging this into the definition of the deformation, one obtains
\begin{equation}
    f'(\lambda) \mathcal{L}^{(0)} = \sqrt{\frac{1}{2} T^{A [\lambda]}_{\ B} T^{B [\lambda]}_{\ A} - \frac{1}{4} \left(T^{A [\lambda]}_{\ A} \right)^2} = f(\lambda) M = f(\lambda) \mathcal{L}^{(0)}.
\end{equation}
Therefore the constraint of $f(\lambda)$
\begin{equation}
    f'(\lambda) = f(\lambda)
\end{equation}
It is worth noting that when $\lambda = 0$, the Lagrangian will degenerate to the seed theory, which indicates that $f(0) = 1$. Thus, one can simply work out the solution of the above equation
\begin{equation}
    f (\lambda) = e^{\lambda}.
\end{equation}
Inserting this back to the factorized formula \eqref{eq:factorize}, one then obtains the deformed data
\begin{equation} \label{eq:L-Boson}
    \mathcal{L}^{[\lambda]} = e^{\lambda} \mathcal{L}^{(0)}, \quad T^{A [\lambda]}_{\ B} = e^{\lambda} T^{A (0)}_{\ B}.
\end{equation}
This indicates that all $a_n$-s are equal to 1. This is a trivial effect to the action since we can rescale it to remove the constant $e^{\lambda}$. Then the deformed correlators defined in \eqref{eq:path<Xn>} are not affected by the $\sqrt{T \overline{T}}$ deformation, namely
\begin{equation} \label{eq:rootcorre}
    \<X_n\>_{[\lambda]}^{\sqrt{T \overline{T}}} = \<X_n\>_{[0]}.
\end{equation}
Finally, we should remark that the result we derived here does not contradict the first-order correction computation because they are from different perspectives. On the one hand, the computation of the correlator \eqref{eq:rootcorre} is the non-perturbative version. It turns out that the $\sqrt{T \overline{T}}$ deformed BMS free scalar model is still the BMS invariant field theory, which is consistent with the property of a non-perturbative marginal deformation. Moreover, since the action is invariant under scaling transformation, the $\sqrt{T \overline{T}}$ deformed free scalar model is the same as its seed theory. On the other hand, the generic first-order correction of the $\sqrt{T \overline{T}}$ discussed in subsection \ref{sec:1strootTTbar} is computed from the perturbative method, which may break the BMS symmetry. So it is normal to use some extra terms. Specifically, while keeping the factor $e^{\lambda}$ in \eqref{eq:L-Boson} to perturbatively compute the first-order correction for the $\sqrt{T \overline{T}}$ deformed free scalar model, the result will be the same as the generic first-order correction of the $\sqrt{T \overline{T}}$ discussed in subsection \ref{sec:1strootTTbar}.

\section{Deforms for free Fermion model} \label{sec:deformFermion}

\subsection{Data of seed theory}

The action of the BMS free fermion model is constructed by the field $\psi_a = (\psi_1, \psi_2)$  as \cite{Hao:2022xhq, Banerjee:2022ocj}
\begin{equation} \label{eq:seedFermion}
    S^{(0)} = \int \td x \td y \mathcal{L}^{(0)}, \quad \mathcal{L}^{(0)} = \psi_1 \p_x \psi_1 - \frac{1}{2} \psi_2 \p_y \psi_1 - \frac{1}{2} \psi_1 \p_y \psi_2,
\end{equation}
with the following equation of motion (EoM)
\begin{equation} \label{eq:FermionEoM}
    \p_y \psi_1 = 0, \quad 2 \p_x \psi_1 = \p_y \psi_2.
\end{equation}
By using the definition \eqref{eq:recursion}, the prototype of stress tensor can be derived as
\begin{equation} \label{eq:prototype}
    \mathcal{T}^{\mu (0)}_{\ \nu} = \frac{\p \mathcal{L}^{(0)}}{\p (\p_{\mu} \psi_a)} \p_{\nu} \psi_a - \delta^{\mu}_{\ \nu} \mathcal{L}^{(0)} = 
    \begin{pmatrix}
        - \psi_1 \p_x \psi_1 & - \frac{1}{2} \psi_2 \p_x \psi_1 - \frac{1}{2} \psi_1 \p_x \psi_2 \\
        \psi_1 \p_y \psi_1 & \frac{1}{2} \psi_2 \p_y \psi_1 + \frac{1}{2} \psi_1 \p_y \psi_2
    \end{pmatrix}.
\end{equation}
This, however, is not in the same form as the standard expression of \eqref{eq:Tcompo}, which is derived from the invariance of the BMS transform. In addition, it is even not a conserved current since \eqref{eq:conservT} is seemingly not satisfied. Fortunately, the standard stress tensor can be derived by plugging the EoM \eqref{eq:FermionEoM} into its prototype \eqref{eq:prototype} as
\begin{equation} \label{eq:TplogEoM}
    T^{\mu (0)}_{\ \nu} = 
    \begin{pmatrix}
        M & T \\
        0 & -M
    \end{pmatrix},
    \quad M = - \psi_1 \p_x \psi_1, \quad T = - \frac{1}{2} \psi_2 \p_x \psi_1 - \frac{1}{2} \psi_1 \p_x \psi_2,
\end{equation}
which satisfies the conservation law \eqref{eq:conservT} and is consistent with the standard form of stress tensor \eqref{eq:Tcompo}. The enlarged symmetry of BMS free Fermion model is triggered by the dilation symmetry $\mathcal{D}'$ of the seed action \eqref{eq:seedFermion}, where \cite{Yu:2022bcp}
\begin{equation} \label{eq:dilation}
    \mathcal{D}': (x, y) \rightarrow (x, D y), \quad (\psi_1, \psi_2) \rightarrow (D^{- \frac{1}{2}} \psi_1, D^{\frac{1}{2}} \psi_2),
\end{equation}
whose Noether currents are
\begin{equation}
    J_{\mathcal{D}'}^{\mu (0)} = T^{\mu (0)}_{\ y} y - J^{\mu}_{(0)}, \quad J^{\mu}_{(0)} = \frac{\p \mathcal{L}^{(0)}}{\p (\p_{\mu} \psi_a)} \mathcal{F}_a, \quad \mathcal{F}_a = \frac{1}{2} (- \psi_1, \psi_2).
\end{equation}
The BMS symmetry is enlarged by a dimension 1 current $J^{\mu}_{(0)}$, whose components are
\begin{equation}
    J^y_{(0)} = J_x^{(0)} = - \frac{1}{2} :\psi_1 \psi_2 :, \quad J^x_{(0)} = - J_y^{(0)} = 0.
\end{equation}
The generators yielded from the stress tensor and $J^{\mu}_{(0)}$ form the BMS Kac-Moody algebra, which is a specific case of \eqref{eq:NLKM}, see the details in \cite{Hao:2022xhq,Yu:2022bcp}.
In BMS free Fermion model, there are 3 kinds of primary fields: identity operators (singlet) with $\Delta = \xi = 0$; Fermion field (multiplet) $\psi = (\psi_1, \psi_2)^T$ with conformal weight $\Delta = \frac{1}{2}$ and boost charge $\xi = 0$; composite operator $P = -2 J^y_{(0)} = : \psi_1 \psi_2 :$ (singlet) with $\Delta = 1, \xi = 0$. With the OPEs between the fields
\begin{equation} \label{eq:FpsiOPE}
    \begin{split}
        & \psi_1(x_1) \psi_1(x_2) \sim 0, \quad \psi_2(x_1,y_1) \psi_2(x_2,y_2) \sim - \frac{2y_{12}}{x_{12}^2}, \\
        & \psi_1(x_1) \psi_2(x_2,y_2) \sim \psi_2(x_1,y_1) \psi_1(x_2) \sim \frac{1}{x_{12}},
    \end{split}
\end{equation}
one can easily check that the OPEs between the currents and the primary operators as
\begin{equation} \label{eq:TpsiOPE}
    \begin{split}
        T(x',y') \psi_1 (x) \sim & \frac{\psi_1(x)}{2(x'-x)^2} + \frac{\p_x \psi_1 (x)}{x'-x}, \\
        T(x',y') \psi_2 (x,y) \sim & \frac{\psi_2 (x,y)}{2 (x'-x)^2} - \frac{2(y'-y)}{(x'-x)^3} \psi_1 (x') \\
        & + \frac{\p_x \psi_2(x,y)}{x'-x} - \frac{y'-y}{(x'-x)^2} \p_y \psi_2 (x,y), \\
        M(x') \psi_1 (x) \sim & 0, \quad M(x') \psi_2 (x,y) \sim \frac{\psi_1 (x)}{(x'-x)^2} + \frac{\p_y \psi_2 (x,y)}{x'-x};
    \end{split}
\end{equation}
\begin{equation} \label{eq:TPOPE}
    \begin{split}
        T(x',y') P(x,y) & \sim \frac{P(x,y)}{(x'-x)^2} + \frac{\p_x P(x,y)}{x'-x} - \frac{y'-y}{(x'-x)^2} \p_y P(x,y), \\
        M(x') P(x,y) & \sim \frac{\p_y P(x,y)}{x'-x};
    \end{split}
\end{equation}
\begin{equation}
    \begin{split}
        P (x_1, y_1) P (x_1, y_1) & \sim \frac{1}{x_{12}^2} - \frac{y_{12}}{x_{12}} \p_{y_2} P(x_2, y_2), \quad P (x_2, y_2) \psi_1 (x_2) \sim \frac{\psi_1 (x_1)}{x_{12}}, \\
        P(x_1, y_1) \psi_2 (x_2, y_2) & \sim - \frac{\psi_2 (x_2)}{x_{12}} + 2 \frac{y_{12}}{x_{12}^2} \psi_1 (x_2) + 2 \frac{y_{12}}{x_{12}} \p_{x_2} \psi_1 (x_2).
    \end{split}
\end{equation}
Comparing with \eqref{eq:WardM}, 
one can easily deduce that $\psi_i$ is a multiplet with weight $\Delta = \frac{1}{2}$ and vanished boost charge while $P$ is a singlet with conformal weight $\Delta = 1$ and vanished boost charge. Moreover, it is easy to verify that OPEs between the components of stress tensor also consist with \eqref{eq:TMOPE} with $c_L = 1$ and $c_M = 0$. 

Similarly, the correlators of the primaries in free Fermion can be derived from the OPEs between the fields
\begin{equation} \label{eq:3ptFermi}
    \begin{split}
        & \< \psi_1(x_1) \psi_2 (x_2, y_2) \>^{(0)} = \frac{1}{x_{12}}, \quad \< \psi_2 (x_1, y_1) \psi_2 (x_2, y_2) \>^{(0)} = - \frac{2 y_{12}}{x_{12}^2}, \\
        & \< \psi_1 (x_1) \psi_2 (x_2, y_2) P (x_3, y_3) \>^{(0)} = - \frac{1}{x_{23} x_{13}}.
    \end{split}
\end{equation}
Other combinations of two-point and three-point functions of primaries are all vanished.

\subsection{$T \overline{T}$ deformation} \label{sec:TTbarFermion}

As derived in \eqref{eq:recursion}, the recursion relation of the deformed Lagrangian is 
\begin{equation} \label{eq:FermionLT}
    \begin{split}
        \mathcal{L}^{(n+1)} = & \frac{1}{2} \sum^n_{i = 0} C^i_n \left( T^{\mu(i)}_{\ \mu} T^{\nu (n-i)}_{\ \nu} - T^{\mu (i)}_{\ \nu} T^{\nu (n-i)}_{\ \mu} \right), \\
        T^{\mu (n)}_{\ \nu} = & \frac{\p \mathcal{L}^{(n)}}{\p (\p_{\mu} \psi_a)} \p_{\nu} \psi_a - \delta^{\mu}_{\ \nu} \mathcal{L}^{(n)}, \quad n = 1, 2, \cdots.
    \end{split}
\end{equation}
With the explicit form of the stress tensor in seed theory $T^{\mu (0)}_{\ \nu}$ derived in \eqref{eq:TplogEoM}, the deformed Lagrangian can be computed order by order. It is easy to verify that the first-order correction of the Lagrangian is $MM$, which is consistent with the discussion in the section \ref{sec:1stTTbarJTaroot}. However, one can immediately obtain that $\mathcal{L}^{(0)}$ vanishes since $MM = (\psi_1 \p_x \psi_1)^2 = 0$. This is because the Grassmann numbers $\psi_1$ and $\p_x \psi_1$ should not appear twice. Therefore, the first-order correction to the components of the stress tensor all vanished
\begin{equation}
    T^{\mu (1)}_{\ \nu} = 0.
\end{equation}
indicating that the second-order correction to the Lagrangian and stress tensor vanish, which can be easily verified as
\begin{equation}
    \begin{split}
        \mathcal{L}^{(2)} & = \frac{1}{2} \left[ T^{\mu (0)}_{\ \mu} T^{\nu (1)}_{\ \nu} + T^{\mu (1)}_{\ \mu} T^{\nu (0)}_{\ \nu} - T^{\mu (0)}_{\ \nu} T^{\nu (1)}_{\ \mu} - T^{\mu (1)}_{\ \nu} T^{\nu (0)}_{\ \mu} \right] = 0 \\
        T^{\mu (2)}_{\ \nu} & = \frac{\p \mathcal{L}^{(2)}}{\p (\p_{\mu} \psi_a)} \p_{\nu} \psi_a - \delta^{\mu}_{\ \nu} \mathcal{L}^{(2)} = 0.
    \end{split}
\end{equation}
By induction on $n$, one can verify that all the correction terms vanish
\begin{equation}
    \mathcal{L}^{(n)} = 0, \quad T^{\mu (n)}_{\ \nu} = 0, \quad \forall n = 1, 2, \cdots,
\end{equation}
causing the Lagrangian and the correlators of BMS free Fermion model to be unchanged through the $T \Bar{T}$ flow, namely
\begin{equation} \label{eq:FermionLX}
    \mathcal{L}^{[\lambda]} = \mathcal{L}^{(0)}, \quad \<X_n\>^{[\lambda]}_{T \overline{T}} = \<X_n\>^{(0)}.
\end{equation}
Until now, we have not used the EoM for the corrected terms. Actually if one imposes the definition of $T^{\mu (n)}_{\ \nu}$ in \eqref{eq:FermionLT} for $T^{\mu (0)}_{\ \nu}$, namely substitute $\mathcal{T}^{\mu (0)}_{\ \nu}$ for $T^{\mu (0)}_{\ \nu}$ 
\begin{equation}
    \begin{split}
        \mathcal{L}^{(1)} 
        = & \frac{1}{2} \left[ \mathcal{T}^{x (0)}_{\ x} \mathcal{T}^{y (0)}_{\ y} + \mathcal{T}^{y (0)}_{\ y} \mathcal{T}^{x (0)}_{\ x} - \mathcal{T}^{x (0)}_{\ y} \mathcal{T}^{y (0)}_{\ x} - \mathcal{T}^{x (0)}_{\ y} \mathcal{T}^{y (0)}_{\ x} \right] \\
        = & \psi_1 \psi_2 \p_x \psi_1 \p_y \psi_1,
    \end{split}
\end{equation}
one can also derive that $T^{\mu (1)}_{\ \nu} = 0$ without the EoM of fields. Then \eqref{eq:FermionLX} can be re-derived by induction on $n$ of $\mathcal{L}^{(n)}$.

The $T \overline{T}$ deformation will change the theory in normal circumstances. However, the result here has shown that the $T \overline{T}$ deformed free Fermion model is a fixed point through the $T \overline{T}$ flow. This is not strange, since the structure of the Fermion model requires that correct terms with multi-$M$ are all zero. Moreover, the deformed data only contain $M$ and are independent of $T$. Therefore all corrections of $T \overline{T}$ are vanished.  

\subsection{$JT_{\mu}$ deformation}

The definition of $JT_{\mu}$ is \eqref{eq:JTL}. With the expansion similar to \eqref{eq:expandLTj}, one can similarly deduce the recursion relation of the deformed Lagrangian as
\begin{equation}
    \mathcal{L}_{\mu_1 \cdots \mu_n \mu_{n+1}} = \sum_{i = 0}^n C_n^i \left[ (J^y)_{\mu_{i+1} \cdots \mu_n} (T^x_{\ \mu_{n+1}})_{\mu_1 \cdots \mu_i} - (J^x)_{\mu_{i+1} \cdots \mu_n} (T^y_{\ \mu_{n+1}})_{\mu_1 \cdots \mu_i} \right],
\end{equation}
where
\begin{equation}
    \begin{split}
        (T^{\alpha}_{\ \mu})_{\mu_1 \cdots \mu_n} & = \frac{\p \mathcal{L}_{\mu_1 \cdots \mu_n}}{\p (\p_{\alpha} \psi_a)} \p_{\mu} \psi_a - \delta^{\alpha}_{\ \mu} \mathcal{L}_{\mu_1 \cdots \mu_n}, \\
        (J^{\beta})_{\mu_{i+1} \cdots \mu_n} & = \frac{1}{2} \frac{\p \mathcal{L}_{\mu_{i+1} \cdots \mu_n}}{\p (\p_{\beta} \psi_2)} \psi_2 - \frac{1}{2} \frac{\p \mathcal{L}_{\mu_{i+1} \cdots \mu_n}}{\p (\p_{\beta} \psi_1)} \psi_1.
    \end{split}
\end{equation}
Similarly, the BMS free Fermion model is also a fixed point for $JT_{\mu}$ flow. To see this, we just need to compute the first-order correction of the Lagrangian. Having discussed in subsection \ref{sec:TTbarFermion}, the on-shell condition should not be implemented while computing the corrections order by order, namely one should substitute the off-shell stress tensor of the seed theory \eqref{eq:prototype} for $T^{\mu (0)}_{\ \nu}$ here, instead of \eqref{eq:TplogEoM}. The first-order correction of the Lagrangian should be $\lambda^{\mu} \mathcal{L}_{\mu} = \lambda^y \mathcal{L}_y + \lambda^x \mathcal{L}_x$, where
\begin{equation}
    \mathcal{L}_y = J^y_{(0)} T^{x (0)}_{\ y} - J^x_{(0)} T^{y (0)}_{\ y} = \frac{1}{4} \psi_1 \psi_2 ( \psi_2 \p_x \psi_1 + \psi_1 \p_x \psi_2 )
\end{equation}
and
\begin{equation}
    \mathcal{L}_x = J^y_{(0)} T^{x (0)}_{\ x} - J^x_{(0)} T^{y (0)}_{\ x} = - \frac{1}{4} \psi_1 \psi_2 ( \psi_2 \p_y \psi_1 + \psi_1 \p_y \psi_2).
\end{equation}
Then one can easily deduce that the above two equations have all vanished since $\psi_1$ or $\psi_2$ appears twice on the right-hand side of both equations. Therefore even without implementing the on-shell condition, the first-order corrections still vanished
\begin{equation}
    \mathcal{L}_y = 0, \quad \mathcal{L}_x = 0.
\end{equation}
Then, by using mathematical induction, one can easily prove that the higher-order corrections have also vanished, namely
\begin{equation}
    \mathcal{L}_{\mu_1 \cdots \mu_n} \equiv 0, \quad \forall n \geq 1.
\end{equation}
This yields that the $JT_{\mu}$ deformed Lagrangian and deformed correlators remain unchanged
\begin{equation}
    \mathcal{L}^{[\lambda]} = \mathcal{L}^{(0)}, \quad \<X_n\>^{[\lambda]}_{JT_{\mu}} = \<X_n\>^{(0)}.
\end{equation}
which means that BMS free Fermion model is also a fixed point through $JT_{\mu}$ flow.

\subsection{Root-$T \overline{T}$ deformation}

\subsubsection{Deformed Lagrangian}

The recursion relation can be derived from \eqref{eq:rootrecursion}, with $T^{\mu (0)}_{\ \nu}$ in \eqref{eq:TplogEoM}
\begin{equation} \label{eq:rootFermi}
    \begin{split}
        \sum_{n = 0}^{\infty} \frac{\lambda^n}{n!} \mathcal{L}^{(n+1)} 
        & = \sqrt{M^2 + \sum_{n = 1}^{\infty} \frac{\lambda^n}{n!} \sum_{i = 0}^n C_n^i \left( \frac{1}{2} T^{A (i)}_{\ B} T^{B (n-i)}_{\ A} - \frac{1}{4} T^{A (i)}_{\ A} T^{B (n - i)}_{\ B} \right)} \\
        T^{\mu (n)}_{\ \nu} & = \frac{\p \mathcal{L}^{(n)}}{\p (\p_{\mu} \psi_a)} \p_{\nu} \psi_a - \delta^{\mu}_{\ \nu} \mathcal{L}^{(n)}, \quad n = 1, 2, \cdots.
    \end{split}
\end{equation}
Note that in the free Fermion case, $MM$ is zero while the $\sqrt{MM} = M$ is formally non-zero. So the first-order correction is
\begin{equation} \label{eq:FermirootL1}
    \mathcal{L}^{(1)} = M = - \psi_1 \p_x \psi_1,
\end{equation}
which consists of the generic discussion in section \ref{sec:1stTTbarJTaroot}. Then the first-order correction of the components of the stress tensor are 
\begin{equation} \label{eq:TcompoFermi}
    T^{y (1)}_{\ y} = - M, \quad T^{x (1)}_{\ y} = - \psi_1 \p_y \psi_1, \quad T^{y (1)}_{\ x} = T^{x (1)}_{\ x} = 0.
\end{equation}
So the second order correction is
\begin{equation} \label{eq:FermirootL2}
    \mathcal{L}^{(2)} = \frac{M^{-1}}{4} \left[ 2 T^{A (0)}_{\ B} T^{B (1)}_{\ A} - T^{A (0)}_{\ A} T^{B (1)}_{\ B} \right] = - \frac{M}{2}.
\end{equation}
One can prove that the corrected terms $\mathcal{L}^{(n)}$ with $n = 1, 2, \cdots$ are all proportional to $\mathcal{L}^{(1)}$ by induction on $n$ with the coefficient $b_n$ for $n$-th order. We have verified that this is true for $n = 1,2$. Suppose this proposition is true for all $n = 1, 2, \cdots, n_0$, namely
\begin{equation}
    \mathcal{L}^{(n)} = b_n M = b_n \mathcal{L}^{(1)}, \quad n = 1, 2, \cdots, n_0.
\end{equation}
So the $n$-th order correction of the stress tensor is proportional to $T^{\mu (1)}_{\ \nu}$
\begin{equation}
    T^{A (n)}_{\ B} = b_n T^{A (1)}_{\ B}, \quad n = 1, 2, \cdots, n_0.
\end{equation}
Then the formula \eqref{eq:rootFermi} can be divided into the following form
\begin{equation}
    \begin{split}
        \sum_{n = 0}^{\infty} \frac{\lambda^n}{n!} \mathcal{L}^{(n+1)} = & \Bigg\{M^2 + \sum_{n = 1}^{n_0} \frac{\lambda^n}{n!} \bigg[ b_n \left( T^{A (0)}_{\ B} T^{B (1)}_{\ A} - \frac{1}{2} T^{A (0)}_{\ A} T^{B (1)}_{\ B} \right) \\
        & + \left( \sum_{i = 1}^{n - 1} \frac{C_n^i}{2} b_i b_{n - i} \right) \left( T^{A (1)}_{\ B} T^{B (1)}_{\ A} - \frac{1}{2} T^{A (1)}_{\ A} T^{B (1)}_{\ B} \right) \bigg] \\
        & + \sum_{n = n_0 + 1}^{\infty} \frac{\lambda^n}{n!} \sum_{i = 0}^n C_n^i \left( \frac{1}{2} T^{A (i)}_{\ B} T^{B (n-i)}_{\ A} - \frac{1}{4} T^{A (i)}_{\ A} T^{B (n - i)}_{\ B} \right) \Bigg\}^{\frac{1}{2}},
    \end{split}
\end{equation}
which can be rewritten as the following form with the expression of $T^{\mu (1)}_{\ \nu}$ and $T^{\mu (0)}_{\ \nu}$ in \eqref{eq:TcompoFermi} and \eqref{eq:Tcompo}
\begin{equation}
    \begin{split}
        \sum_{n = 0}^{\infty} \frac{\lambda^n}{n!} \mathcal{L}^{(n+1)} = & \Bigg\{M^2 \left[ 1 + \sum_{n = 1}^{n_0} \frac{\lambda^n}{n!} \left( \sum_{i = 1}^{n - 1} \frac{C_n^i}{4} b_i b_{n - i} - b_n \right) \right] \\
        & + \sum_{n = n_0 + 1}^{\infty} \frac{\lambda^n}{n!} \sum_{i = 0}^n C_n^i \left( \frac{1}{2} T^{A (i)}_{\ B} T^{B (n-i)}_{\ A} - \frac{1}{4} T^{A (i)}_{\ A} T^{B (n - i)}_{\ B} \right) \Bigg\}^{\frac{1}{2}}.
    \end{split}
\end{equation}
After expanding the square root near $M$ by the power of $\lambda$ and reading off the coefficient in front of the $\lambda^{n_0}$ in the right-hand side, which comes from the first line of the above equation, one can deduce that the $n_0 + 1$-th order correction of the deformed Lagrangian at $n_0 + 1$-th order is also proportional to $M (= \mathcal{L}^{(1)})$. Therefore the all-order corrected Lagrangian and stress tensor can be expressed as
\begin{equation} \label{eq:LFrootTTbar}
    \mathcal{L}^{[\lambda]} = \mathcal{L}^{(0)} + g(\lambda) M, \quad T^{A [\lambda]}_{\ B} = T^{A (0)}_{\ B} + g(\lambda) T^{A (1)}_{\ B}.
\end{equation}
Plugging them into the definition of the root-$T \overline{T}$ deformation \eqref{eq:rootdefine} and \eqref{eq:R}, one can derive a constraint for the function $g(\lambda)$ as
\begin{equation}
    g'(\lambda) = \frac{1}{2} |g(\lambda) - 2|
\end{equation}
Similarly, the Lagrangian will degenerate to the seed theory $\mathcal{L}^{(0)}$, so $g(0) = 0$. Together with the constraints \eqref{eq:FermirootL1} and \eqref{eq:FermirootL2}, the function $g(\lambda)$ can be fixed as
\begin{equation} \label{eq:glambda}
    g(\lambda) = 2 - 2 e^{- \frac{\lambda}{2}}.
\end{equation}
Unfortunately, if we substitute $\mathcal{T}^{\mu (0)}_{\ \nu}$ for $T^{\mu (0)}_{\ \nu}$, as we did in subsection \ref{sec:TTbarFermion}, the off-shell terms within the square root will not form a perfect square. Consequently, the definition of the $\sqrt{T \overline{T}}$ deformation for the BMS fermion model becomes ill-defined, as the square root of certain Grassmann numbers lacks a clear definition. Before proceeding further, it is worth noting that the flow of the fields is not taken into account, which does not pose any issues within the perturbative approach.

\subsubsection{Deformed correlator}

The deformed correlators can be computed by path integral as
\begin{equation}
    \begin{split}
        \<X_n\>^{\sqrt{T \overline{T}}}_{[\lambda]} & = \left\< e^{- g (\lambda) \int \td x \td y M(x)} X_n \right\> \\
        & = \text{exp} \left[ 2 \pi i g (\lambda) \sum_k y_k \p_{y_k} \right] \<X_n\> \\
        & = \text{exp} \left[4 \pi i \left( 1- e^{- \frac{\lambda}{2}} \right) \sum_k y_k \p_{y_k} \right] \<X_n\>,
    \end{split}
\end{equation}
where we used the Ward identity \eqref{eq:WardM} in the second line, and the integral here is the same as we did in subsection \ref{sec:1strootTTbar}. Now the rest of this sub-subsection is to manifest the extra poles generated from the deformation by using the data of the seed theory derived in \eqref{eq:3ptFermi}, instead of leaving the derivative $\p_{y_k}$-s here. Apart from $\<\psi_2 \psi_2\>$ in \eqref{eq:3ptFermi}, others are all independent on $y_k$, indicating that
\begin{equation}
    \begin{split}
        & \< \psi_1(x_1) \psi_2 (x_2, y_2) \>_{[\lambda]}^{\sqrt{T \overline{T}}} = \< \psi_1(x_1) \psi_2 (x_2, y_2) \>^{(0)}, \\
        & \< \psi_1 (x_1) \psi_2 (x_2, y_2) P (x_3, y_3) \>_{[\lambda]}^{\sqrt{T \overline{T}}} = \< \psi_1 (x_1) \psi_2 (x_2, y_2) P (x_3, y_3) \>^{(0)}.
    \end{split}
\end{equation}
Since $\< \psi_2 (x_1, y_1) \psi_2 (x_2, y_2) \>^{(0)}$ is proportional to $y_{12}$, together with the fact that $\sum_k y_k \p_{y_k}$ is the identity operator of $y_{12}$, one can easily verify that
\begin{equation}
    \< \psi_2 (x_1, y_1) \psi_2 (x_2, y_2) \>_{[\lambda]}^{\sqrt{T \overline{T}}} = \text{exp} \left[4 \pi i \left( 1- e^{- \frac{\lambda}{2}} \right) \right] \< \psi_2 (x_1, y_1) \psi_2 (x_2, y_2) \>^{(0)}.
\end{equation}
Therefore, the only impact of the $\sqrt{T \overline{T}}$ deformation on $\< \psi_2 (x_1, y_1) \psi_2 (x_2, y_2) \>$ is a mere phase factor. This assures us that the BMS symmetries of the correlators in the BMS-free Fermion case remain intact despite the deformation. This phenomenon can be observed explicitly from the Lagrangian perspective by redefining the field $\psi_1$ as $\psi'_1 = (1 - g(\lambda)) \psi_1 = (2 e^{- \frac{\lambda}{2}} - 1) \psi_1$. In this redefinition, combined with \eqref{eq:FermirootL1}, \eqref{eq:glambda} and \eqref{eq:seedFermion}, the $\sqrt{T \overline{T}}$ deformed Lagrangian \eqref{eq:LFrootTTbar} can be expressed as
\begin{equation}
    \mathcal{L}^{[\lambda]} = \frac{1}{2 e^{- \frac{\lambda}{2}} - 1} \left( \psi'_1 \p_x \psi'_1 - \frac{1}{2} \psi_2 \p_y \psi'_1 - \frac{1}{2} \psi'_1 \p_y \psi_2 \right).
\end{equation}
This formulation is essentially a rescaling of the undeformed Lagrangian \eqref{eq:seedFermion} with field $(\psi'_1, \psi_2)$, which can be interpreted as a dilation transformation with $D = 2 e^{- \frac{\lambda}{2}} - 1$ as defined in \eqref{eq:dilation}. Consequently, the deformed action remains unaffected since the factor $\frac{1}{2 e^{- \frac{\lambda}{2}} - 1}$ can always be absorbed through coordinates rescaling, which is similar as the discussion in subsection \ref{sec:rootscalar}. This highlights that the BMS symmetries of the $\sqrt{T \overline{T}}$ deformed correlator remains preserved. This outcome is expected because the $\sqrt{T \overline{T}}$ deformation is a marginal deformation that preserves the original symmetries of the seed theory. As a result, the $\sqrt{T \overline{T}}$ deformed BMS-free Fermion model can be considered a well-defined marginal deformed theory that still qualifies as a BMSFT even after deformation.

In conclusion, as emphasized in subsection \ref{sec:rootscalar}, it is important to reiterate that the deformed correlators considered here encompass all order corrections and are cast into closed forms. This non-perturbative approach differs significantly from the perturbative method discussed in section \ref{sec:1stTTbarJTaroot}. However, it is crucial to note that this disparity does not imply a contradiction. Perturbative and non-perturbative methods operate at different levels, and upon closer examination, it becomes evident that the first-order corrections obtained from expanding the aforementioned results by power of $\lambda$ are identical to those derived from the perturbative method in section \ref{sec:1stTTbarJTaroot}. Notably, the non-perturbative approach offers greater precision and comprehensiveness compared to its perturbative counterpart. Consequently, the perturbative method runs the risk of compromising the symmetries inherent in the original theory.

\section{Conclusion}

In this paper, we introduce various types of irrelevant and marginal deformations in the BMSFT to evaluate the several types of action and lowest-order corrections to correlation functions. Firstly, we define these irrelevant and marginal deformations properly which is non-lorentize type of deformation. Based on the deformations, we apply the standard perturbative field theory approach to analyze the universal first-order corrections to the correlation functions of seed theories, which, based on our analysis, are only factorized for two-point and three-point functions consisting of singlet primary operators. In addition, we also investigate the flow effects of the deformations by calculating the higher order corrections for some specific case, e.g., free BMS Boson and Fermion theories, since the first-order corrections do not flow the seed theory while the higher order corrections depend on different seed theories. Particularly, we provide the all-order corrected Lagrangian for the deformations for these two cases, and compute the higher-order corrections of the deformed correlation functions systematically. As the classification of the RG, the irrelevant deformation might flow the seed theory to different theories while the well-defined marginal deformations will not. Specifically, the irrelevant $T \overline{T}$ and $JT_{\mu}$ deformations will indeed flow the local BMS free scalar theory to the non-local, and string-like theories, which can be observed from both classic levels, namely the \textit{all-order} corrected Lagrangian
\begin{equation}
    \mathcal{L}^{[\lambda] \text{scalar}}_{T \overline{T}} = \frac{\sqrt{4 \lambda (\p_y \phi)^2 +1}-1}{2 \lambda}, \quad \mathcal{L}^{[\lambda] \text{scalar}}_{JT_{\mu}} = \frac{1 - 2 \lambda \p_y \phi - \sqrt{1 - 4 \lambda \p_y \phi }}{2 \lambda ^2},
\end{equation}
and the quantum level, namely the higher-order corrections of the deformed correlators. However, the $T \overline{T}$ or $JT_{\mu}$ deformed free BMS Fermion model does not have any corrections for both classic (Lagrangian) and quantum (Correlator) level, because of the Grassmann structure of the Fermion. Besides, the $\sqrt{T \overline{T}}$ deformation for free BMS scalar and Fermion are well-defined and exact marginal deformation since they are essentially scale transforms for the Lagrangian
\begin{equation}
    \begin{split}
        \mathcal{L}^{[\lambda] \text{scalar}}_{\sqrt{T \overline{T}}} & = e^{\lambda} \mathcal{L}^{[0] \text{scalar}}_{\sqrt{T \overline{T}}}, \\
        \mathcal{L}^{[\lambda] \text{fermion}}_{\sqrt{T \overline{T}}} & = \mathcal{L}^{(0)} + 2 \left( 1 - e^{- \frac{\lambda}{2}} \right) M \\
        & = \frac{1}{2 e^{- \frac{\lambda}{2}} - 1} \mathcal{L}^{[0] \text{fermion}}_{\sqrt{T \overline{T}}} [\psi'_1, \psi_2],
    \end{split}
\end{equation}
which preserves the original BMS symmetries since the actions are scale invariant and the scale factors can be absorbed by coordinates rescaling.

\section*{Acknowledgments}

The authors thank Bin Chen, Feng Hao, Pujian Mao, Hao Ouyang, and Xiyang Ran, Hongan Zeng, Yu-Xuan Zhang, Zi-Xuan Zhao for their valuable discussions and comments. We are also grateful to the anonymous reviewer for pointing out that the Lagrangian of $\sqrt{T \overline{T}}$ deformed free Fermion can be absorbed by appropriate scalings of the coordinates, which resolves our confusion in the last version. This work is partly supported by the National Natural Science Foundation of China under Grant No.~12075101, No.~12235016. S.H. is grateful for financial support from the Fundamental Research Funds for the Central Universities and the Max Planck Partner Group.

\appendix

\section{Seed BMSFT from UR limit} \label{sec:URseed}

\subsection{Algebras from UR limit} \label{sec:URAlgebera}

In this subsection, we discuss the BMS algebra \eqref{eq:BMFTA} and NLKM algebra \eqref{eq:NLKM}. The BMS algebra can be derived from Virasoro algebra by implementing the UR limit. For generators in 2D CFT 
\begin{equation}
    \mathcal{L}_n = - z^{n+1} \p_z, \quad \mathcal{L}_n = - \Bar{z}^{n+1} \p_{\Bar{z}}, \quad (z, \Bar{z}) = (x+iy, x-iy)
\end{equation}
which satisfies the Virasoro algebra after the central extension
\begin{equation}
    \begin{split}
        [\mathcal{L}_n, \mathcal{L}_m] = & (n - m) \mathcal{L}_{n+m} + \frac{c}{12} n (n^2 - 1) \delta_{n+m,0}, \\
        [\Bar{\mathcal{L}}_n, \Bar{\mathcal{L}}_m] = & (n - m) \Bar{\mathcal{L}}_{n+m} + \frac{\Bar{c}}{12} n (n^2 - 1) \delta_{n+m,0}, \\
        [\mathcal{L}_n, \Bar{\mathcal{L}}_m] = & 0.
    \end{split}
\end{equation}
By choosing the following UR limit
\begin{equation} \label{eq:URxy}
    y \rightarrow \epsilon y, \quad x \rightarrow x, \quad \epsilon \rightarrow 0,
\end{equation}
together with
\begin{equation} \label{eq:URVirasoro}
    L_n = \lim_{\epsilon \rightarrow 0} (\mathcal{L}_n - \Bar{\mathcal{L}}_{-n}), \quad M_n = \lim_{\epsilon \rightarrow 0} \epsilon (\mathcal{L}_n + \Bar{\mathcal{L}}_{-n}).
\end{equation}
\begin{equation} \label{eq:URc}
    c_L = \lim_{\epsilon \rightarrow 0} (c - \Bar{c}), \quad c_M = \lim_{\epsilon \rightarrow 0} \epsilon (c + \Bar{c}),
\end{equation}
the Virasoro algebra will then recover \eqref{eq:BMFTA}.

Similarly, the NLKM can also be derived from the UR limit of the Virasoro Kac-Moody algebra. The holomorphic part of the Virasoro Kac-Moody algebra is
\begin{equation}
    \begin{split}
        [\mathcal{L}_n, \mathcal{L}_m] & = (m - n) \mathcal{L}_{m+n} + \frac{c}{12} (m^3 - m) \delta_{m+n, 0}, \quad [\mathcal{L}_m, j^a_n] = - n j^a_{m+n}, \\
        [j^a_m, j^b_n] & = i f^{abc} j^c_{m+n} + m k \delta_{m+n, 0} \delta^{ab}.
    \end{split}
\end{equation}
where the anti-holomorphic part is similar. By choosing the UR limit as \eqref{eq:URxy}, \eqref{eq:URVirasoro}, and \eqref{eq:URc}, together with
\begin{equation}
    J_m^a = j^a_m + \Bar{j}^a_{-m}, \quad K_m^a = \epsilon (j^a_m - \Bar{j}^a_{-m})
\end{equation}
and
\begin{equation}
    F^{abc} = \frac{1}{2} (f^{abc} + \Bar{f}^{abc}), \quad G^{abc} = \frac{1}{2 \epsilon} (f^{abc} + \Bar{f}^{abc}),
\end{equation}
\begin{equation}
    k_1 = k - \Bar{k}, \quad k_2 = \epsilon (k + \Bar{k}),
\end{equation}
the NLKM \eqref{eq:NLKM} will be explicitly re-derived.

\subsection{OPEs from UR limit} \label{sec:UROPE}

In this subsection, we discuss the Ward identities of the stress tensor in BMSFT and its OPE with the primary operators. We also need to give singlet and the multiplet result.

\paragraph{Singlets} The OPE between 2 operators is related to their commutators from the radial quantization
\begin{equation} \label{eq:[A,B]}
    \begin{split}
        & A = \oint a(z) dz, \quad B = \oint b(w) dw \\
        & [A,B] = \oint_0 dw \oint_w dz a(z) b(w), \quad [A, b(w)] = \oint_w dz a(z) b(w).
    \end{split}
\end{equation}
The OPEs between the components of the stress tensor can be derived from \eqref{eq:BMFTA} as
\begin{equation}
    \begin{split}
        L(x') L(x) \sim & \frac{c_L}{2(x'-x)^4} + \frac{2 L(x)}{(x'-x)^2} + \frac{\p_x L(x)}{x'-x}, \quad M(x') M(x) \sim 0,\\
        L(x') M(x) \sim & M(x') L(x) \sim \frac{c_M}{2(x'-x)^4} + \frac{2 M(x)}{(x'-x)^2} + \frac{\p_x M(x)}{x'-x}.
    \end{split}
\end{equation}
The OPE between the singlet primary operators $\mathcal{O}$ and the stress tensor can be derived from the UR limit. In CFT, we have
\begin{equation}
    [\mathcal{L}_n, \mathcal{O}] = \Big((n+1) h z^n + z^{n+1} \p_z \Big) \mathcal{O}, \quad [\bar{\mathcal{L}}_n, \mathcal{O}] = \Big((n+1) \Bar{h} \bar{z}^n + \bar{z}^{n+1} \p_{\bar{z}} \Big) \mathcal{O}.
\end{equation}
By using UR limit \eqref{eq:URVirasoro}, we obtain
\begin{align}
    [L_n, \mathcal{O}(x,y)] = & \Big[x^{n+1} \p_x + (n+1) x^n y \p_y + (n+1) x^n \Delta + n(n+1) x^{n-1} y \xi \Big] \mathcal{O}(x,y), \notag\\
    [M_n, \mathcal{O}(x,y)] = & \Big[x^{n+1} \p_y + (n+1) x^n \xi \Big] \mathcal{O}(x,y), \quad n \geq -1. \label{eq:primarytrans}
\end{align}
Then the OPE between the components of stress tensor and $\mathcal{O}$ can be easily derived from \eqref{eq:[A,B]} as
\begin{align} 
    T(x,y) \mathcal{O}_k (x_k,y_k) \sim & \frac{\Delta_k \mathcal{O}_k}{(x-x_k)^2} - \frac{2(y - y_k) \xi_k \mathcal{O}_k}{(x-x_k)^3} + \frac{\p_{x_k} \mathcal{O}_k}{x-x_k} - \frac{(y-y_k) \p_{y_k} \mathcal{O}_k}{(x-x_k)^2} \notag \\
    M(x) \mathcal{O}_k (x_k,y_k) \sim & \frac{\xi_k \mathcal{O}_k}{(x-x_k)^2} + \frac{\p_{y_k} \mathcal{O}_k}{x-x_k}. \label{eq:OPEsingle}
\end{align}
We can only roughly see the pole structure by using the UR limit method. More precisely, the pole structure can be observed by using the standard path integral method, which is equivalent to substitute $\Delta \Tilde{x}_k = x - x_k - i \varepsilon (y - y_k)$ for $x - x_k, \ 0 < \varepsilon \ll 1$, see details in \cite{Saha:2022gjw}.

\paragraph{Multiplets} Similarly, for multiplets, we have
\begin{equation} \label{eq:OPEMultip}
    \begin{split}
        T(x,y) O_{ia}(x_i,y_i) \sim & \frac{\Delta O_{ia}}{(\Delta \Tilde{x}_i)^2} - \frac{2(y - y_i) (\boldsymbol{\xi}_i \cdot O_i)_a}{(\Delta \Tilde{x}_i)^3} + \frac{\p_{x_i} O_{ia}}{\Delta \Tilde{x}_i} - \frac{(y-y_i) \p_{y_i} O_{ia}}{(\Delta \Tilde{x}_i)^2} \\
        M(x) O_{ia}(x_i,y_i) \sim & \frac{(\boldsymbol{\xi}_i \cdot O_i)_a}{(\Delta \Tilde{x}_i)^2} + \frac{\p_{y_i} O_{ia}}{\Delta \Tilde{x}_i}.
    \end{split}
\end{equation}

\section{Deformations from UR limit} \label{sec:DeformUR}

In this appendix, we give some insights that the deformations for BMSFT can also be derived from UR limit of the deformations for CFT. Without losing generality and keeping simplicity, we will mainly focus on the discussion of $T \overline{T}$ deformation. Suppose each order of $T \overline{T}$ deformed relativistic CFTs are consisted of stress tensors $T_{zz}^{(0)}, T_{\Bar{z} \Bar{z}}^{(0)}$ in the seed theories. The UR limit will link the coordinate in relativistic CFT $(z, \Bar{z})$ and that in BMSFT $(x,y)$ as
\begin{equation} \label{eq:xyzzbar}
    z = x + \epsilon y + O(\epsilon^2), \quad \Bar{z} = x - \epsilon y + O(\epsilon^2).
\end{equation}
The $\epsilon$ here is equivalent to $i \varepsilon$ in footnote \ref{fn:UR}. The relation between volume element is
\begin{equation}
    - 2 \epsilon \td x \td y = \td^2 z.
\end{equation}
Note that $T_{zz}^{(0)}$ and $T_{\Bar{z} \Bar{z}}^{(0)}$ can be expanded by Virasoro generators as
\begin{equation}
    T_{zz}^{(0)}(z) = \sum_{n \in \mathbb{Z}} z^{-n-2} \mathcal{L}_n, \quad T_{\Bar{z} \Bar{z}}^{(0)}(\Bar{z}) = \sum_{n \in \mathbb{Z}} \Bar{z}^{-n-2} \bar{\mathcal{L}}_n.
\end{equation}
Then, by using \eqref{eq:URVirasoro}, \eqref{eq:xyzzbar} and \eqref{eq:modexpanLM}, one obtains
\begin{equation}
    \begin{split}
        & T_{zz}^{(0)} + T_{\Bar{z} \Bar{z}}^{(0)} = \sum_n (z^{-n-2} \mathcal{L}_n + \Bar{z}^{-n-2} \bar{\mathcal{L}}_n ) \\
        = & \sum_n x^{-n-2} \left[\left( 1 - (n+2) \frac{y}{x} \epsilon \right) \mathcal{L}_n + \left( 1 + (n+2) \frac{y}{x} \epsilon \right) \bar{\mathcal{L}}_n \right] + \mathcal{O} (\epsilon^2) \\
        = & L(x) + y \p_x M(x) + \mathcal{O}(\epsilon^2) = T(x,y).
    \end{split}
\end{equation}
Similarly,
\begin{equation}
    \begin{split}
        & T_{zz}^{(0)} - T_{\Bar{z}^{(0)} \Bar{z}} = \sum_n (z^{-n-2} \mathcal{L}_n - \bar{z}^{-n-2} \bar{\mathcal{L}}_n ) \\
        = & \sum_n x^{-n-2} \left[\left( 1 - (n+2) \frac{y}{x} \epsilon \right) \mathcal{L}_n - \left( 1 + (n+2) \frac{y}{x} \epsilon \right) \bar{\mathcal{L}}_n \right] + \mathcal{O}(\epsilon^2) = \frac{1}{\epsilon} M(x).
    \end{split}
\end{equation}
Therefore, the UR limit yields
\begin{equation} \label{eq:T->M}
    \epsilon T_{zz}^{(0)} = - \epsilon T_{\Bar{z} \Bar{z}}^{(0)} = \frac{M}{2}.
\end{equation}
So we just need to substitute $M$ for $T_{zz}^{(0)}$ and $T_{\Bar{z} \Bar{z}}^{(0)}$ for BMSFT, which indicates that the correction terms for $T \overline{T}$ deformed BMSFT consist entirely of $M$ raised to different powers. This is consistent with our proposal in the main text. 

Specifically, the $T \overline{T}$ deformed scalar model for BMSFT and relativistic CFT are exactly associated with each other by UR limit. In the classic level, the action of the un-deformed relativistic free scalar model is
\begin{equation}
    S^{(0)}_{\text{relativistic}} = - \frac{1}{2} \int \td^2 z \p_z \phi_c \p_{\Bar{z}} \phi_c.
\end{equation}
The corresponding stress tensor in the relativistic seed theory is
\begin{equation}
    T^{(0)}_{zz} = (\p_z \phi_c)^2, \quad T^{(0)}_{\Bar{z} \Bar{z}} = (\p_{\Bar{z}} \phi_c)^2, \quad T^{(0)}_{z \Bar{z}} = 0.
\end{equation}
Together with the definition of $M$ in \eqref{eq:T0BMSFT} and the relation \eqref{eq:T->M}, one obtains the relation between the scalar fields in relativistic CFT and BMSFT
\begin{equation}
    \phi_c = \epsilon^{\frac{1}{2}} \phi,
\end{equation}
which is the same as the rescaling in \cite{Banerjee:2022ime, Hao:2021urq}. The deformed action in relativistic CFT is \cite{Cavaglia:2016oda}
\begin{equation}
    S^{[\lambda']}_{\text{relativistic}} = - \frac{1}{2} \int \td^2 z \mathcal{L}^{[\lambda']}_{\text{relativistic}}, \quad \mathcal{L}^{[\lambda']}_{\text{relativistic}} = \frac{\sqrt{4 \lambda' \p_z \phi_c \p_{\Bar{z}} \phi_c + 1} - 1}{2 \lambda'}
\end{equation}
Then, by taking UR limit, the action becomes
\begin{equation}
    S^{[\lambda']}_{\text{relativistic}} \Big|_{\text{UR}} \rightarrow \int \td x \td y \frac{\sqrt{4 \frac{\lambda}{\epsilon} (\p_y \phi_c)^2 + 1} - 1}{2 \lambda} = \int \td x \td y \frac{\sqrt{4 \lambda (\p_y \phi)^2 + 1} - 1}{2 \lambda},
\end{equation}
where $\lambda = \lambda'/ \epsilon$ is the coupling constant of deformed BMSFT. Therefore, at the classic level, the deformed free scalar model of BMSFT can also be derived from that of relativistic CFT by taking the UR limit. In the quantum level, the $T \overline{T}$ deformed correlators have been discussed in \cite{He:2019vzf, Cardy:2019qao} if the seed theory is relativistic CFT \footnote{Or, one can use the deformed partition function of relativistic free Boson in \cite{He:2020cxp} to derive the deformed correlators of relativistic CFT.}. One can easily check that the quantum corrections of the correlators of the relativistic free scalar correlators will fall off to the BMS free scalar derived in \eqref{eq:XnTTbar12}.

\section{Integral scheme} \label{sec:Integral}

This appendix develops a scheme to work out the integral proposed in \eqref{eq:Integral}. Through the analytical continuation, the integral of $x$ for $a_k >0$ can be extended to a contour integral surrounding the upper half plane with the anticlockwise direction. Notice that the contour integral is trivial while considering all of the poles simultaneously
\begin{equation} \label{eq:oint = 0}
    \oint_{x_1 \sim x_n} \frac{\td x}{\prod_{k = 1}^n (\widetilde{x} - \widetilde{x}_k)^{a_k}} = 0,
\end{equation}
which means that the poles should be placed in different half-plane, or the result will be zero, where the subscript $x_1 \sim x_n$ under the “$\oint$” denotes the residue needs to be computed. Moreover, the extra term $i \epsilon (y - y_k)$ inside the pole of $\widetilde{x}$ requires that the integral of $y$ should be divided by $(y_k, y_{k+1})$ with $k = 1, \cdots, n-1$, because, for $y>y_k$, the pole $x_k$ is in the upper half plane, which is inside of the contour of $x$, while for $y<y_k$, the pole $x_k$ is in the lower half-plane, which is outside of the contour of $x$. After the range of $y$ has already been divided, the extra term $i \epsilon (y - y_k)$ can be removed safely by $\epsilon \rightarrow 0$ since it will no longer contribute to the integral of $x$. Therefore the area integral can be rewritten as
\begin{equation} \label{eq:intdydx}
    \begin{split}
        \mathcal{I}^f_{a_1 \cdots a_n} & = \int^{\infty}_{- \infty} \td y f(y - y_i) \int^{\infty}_{- \infty} \frac{\td x}{\prod_{k = 1}^n (\Delta \Tilde{x}_k)^{a_k}} \bigg|_{\epsilon \rightarrow 0} \\
        & = \sum_{j = 1}^{n-1} \int^{y_{j+1}}_{y_j} \td y f(y - y_i) \oint_{x_{j+1} \sim x_n} \frac{\td x}{\prod_{k = 1}^n (x - x_k)^{a_k}} \\
        & = - \sum_{j = 1}^{n-1} \int^{y_{j+1}}_{y_j} \td y f(y - y_i) \oint_{x_1 \sim x_j} \frac{\td x}{\prod_{k = 1}^n (x - x_k)^{a_k}}.
    \end{split}
\end{equation}
This can be rewritten more beautifully, see the following Lemma and its corollary
\begin{lemma}
    Suppose $y_1 < y_2 < \cdots < y_n$, and $(\cdot)$ denotes the pole structure without $i \epsilon (y - y_k)$ for simplicity. Then \eqref{eq:intdydx} can be rewritten as
    \begin{equation} \label{eq:attachn}
        \int \td y \td x (\cdot) = - \sum_{j = 1}^{n-1} \int^{y_n}_{y_j} \td y \oint_{x_j} \td x (\cdot).
    \end{equation}
\end{lemma}
\begin{proof}
    This can be easily proved by combining the contour integrals with the same pole $x_j$ together.
\end{proof}

\begin{corollary}
    By using \eqref{eq:oint = 0}, we can easily prove that for any $k \in \{1, \cdots, n\}$, \eqref{eq:attachn} can be rewritten as
    \begin{equation} \label{eq:attachk}
        \int \td x \td y (\cdot) = - \sum_{j \neq k} \int^{y_k}_{y_j} \td y \oint_{x_j} \td x (\cdot),
    \end{equation}
\end{corollary}

\begin{proof}
    We divide \eqref{eq:intdydx} into $j \geq k$ and $j < k$
    \begin{align*}
        \int dx dy (\cdot) & = - \left(\sum_{j = k}^{n-1} + \sum_{j = 1}^{k-1} \right) \int^{y_{j+1}}_{y_j} dy \oint_{x_1 \sim x_j} dx (\cdot) \\
        & = \sum_{j = k}^{n-1} \int^{y_{j+1}}_{y_j} dy \oint_{x_{j+1} \sim x_n} dx (\cdot) - \sum_{j = 1}^{k-1} \int^{y_k}_{y_j} dy \oint_{x_j} dx (\cdot) \\
        & = \sum_{j = k}^{n-1} \int^{y_{j+1}}_{y_k} dy \oint_{x_{j+1}} dx (\cdot) - \sum_{j = 1}^{k-1} \int^{y_k}_{y_j} dy \oint_{x_j} dx (\cdot) \\
        & = \sum_{j = k+1}^{n} \int^{y_j}_{y_k} dy \oint_{x_j} dx (\cdot) - \sum_{j = 1}^{k-1} \int^{y_k}_{y_j} dy \oint_{x_j} dx (\cdot) \\
        & = - \sum_{j = k+1}^{n} \int_{y_j}^{y_k} dy \oint_{x_j} dx (\cdot) - \sum_{j = 1}^{k-1} \int^{y_k}_{y_j} dy \oint_{x_j} dx (\cdot),
    \end{align*}
    where we used \eqref{eq:oint = 0} in the second step.
\end{proof}
\noindent Since $y_k$ can be chosen as any operator there is no need to constrain $y_1 < \cdots < y_n$ while using \eqref{eq:attachk}. Therefore \eqref{eq:intdydx} is
\begin{equation}
    \mathcal{I}_{a_1 \cdots a_n}^f = \sum_{j = 1}^n \int^{y_k}_{y_j} \td y f(y-y_i) \oint_{x_j} \frac{\td x}{\prod^n_{i=1} (x-x_i)^{a_i}}
\end{equation}

Then the integrals can be easily computed. In $T \overline{T}, \sqrt{T \overline{T}}$ case, one will meet $f = 1$ cases, namely to compute
\begin{equation}
    \mathcal{I}_{a_1 \cdots a_n} = \frac{1}{(a_1 -1) !} \cdots \frac{1}{(a_n - 1) !} \p_{x_1} \cdots \p_{x_n} \mathcal{I}_{x_1 \cdots x_n}.
\end{equation}
The new quantity $\mathcal{I}^f_{x_1 \cdots x_n}$ is introduced for simplicity
\begin{equation}
    \mathcal{I}_{x_1 \cdots x_n} = - \sum_{j = 1}^n y_{ij} \oint_{x_i} \frac{\td x }{\prod_{k = 1}^n (x - x_k)},
\end{equation}
where we choose the reference point as $y_k = y_i$. This kind of integrals can be computed by using residue theorem as
\begin{equation}
    \mathcal{I}_{x_1 \cdots x_n} = 2 \pi i \frac{y_{ji}}{\prod_{k (\neq i)} x_{ik}}
\end{equation}
Here we present the integrals that appeared in the main text 
\begin{equation}
    \begin{split}
        \mathcal{I}_{4040} & = \left( y_{14} \oint_{x_1} + y_{24} \oint_{x_2} + y_{34} \oint_{x_3} \right) \frac{dx}{(x-x_1)^4 (x-x_3)^4} \\
        & = \frac{2 \pi i}{6} \left( y_{14} \p_{x_1}^3 + y_{34} \p_{x_3}^3 \right) \frac{1}{x_{13}^4} = \frac{2 \pi i}{6} y_{13} \p_{x_1}^3 \frac{1}{x_{13}^4} = - 40 \pi i \frac{y_{13}}{x_{13}^7},
    \end{split}
\end{equation}
\begin{equation}
    \mathcal{I}_{1313} = 2 \pi i \left[ \frac{y_{14}}{x_{12}^3 x_{13} x_{14}^3} + \frac{y_{34}}{x_{32}^3 x_{31} x_{34}^3} - y_{24} \frac{ 6 x_{12}^2 x_{23}^2 + x_{24}^2 x_{13}^2 + 3 x_{21} x_{23} x_{42} ( x_{12} + x_{32} + x_{42} ) }{ x_{12}^3 x_{23}^3 x_{24}^5 } \right],
\end{equation}
\begin{equation}
    \begin{split}
        \mathcal{I}_{2222} = & y_{14} \oint_{x_1} \frac{dx}{\prod^4_{i = 1} (x-x_i)^2} + (1 \leftrightarrow 2) + (1 \leftrightarrow 3) = \p_{x_1} \frac{2 \pi i y_{14}}{x_{12}^2 x_{13}^2 x_{14}^2} + (1 \leftrightarrow 2) + (1 \leftrightarrow 3) \\
        = & - 4 \pi i \left[ y_{14} \left( \frac{1}{x_{12}^3 x_{13}^2 x_{14}^2} + \frac{1}{x_{12}^2 x_{13}^3 x_{14}^2} + \frac{1}{x_{12}^2 x_{13}^2 x_{14}^3} \right) \right] + (1 \leftrightarrow 2) + (1 \leftrightarrow 3),
    \end{split}
\end{equation}
\begin{equation}
    \mathcal{I}_{3030} = \left( y_{14} \oint_{x_1} + y_{24} \oint_{x_2} + y_{34} \oint_{x_3} \right) \frac{dx}{(x-x_1)^3 (x-x_3)^3} = \p_{x_1}^2 \frac{ \pi i y_{13}}{x_{13}^3} = 12 \pi i \frac{y_{13}}{x_{13}^5},
\end{equation}
\begin{equation}
    \mathcal{I}_{2020} = \left( y_{14} \oint_{x_1} + y_{24} \oint_{x_2} + y_{34} \oint_{x_3} \right) \frac{dx}{(x-x_1)^2 (x-x_3)^2} = \p_{x_1} \frac{ 2 \pi i y_{13}}{x_{13}^2} = - 4 \pi i \frac{y_{13}}{x_{13}^3},
\end{equation}
\begin{equation}
    \mathcal{I}_{1212} = 2 \pi i \left[ \frac{y_{14}}{x_{12}^2 x_{13} x_{14}^2} + \frac{y_{34}}{x_{32}^2 x_{31} x_{34}^2} + y_{24} \frac{2 x_{12} x_{23} + x_{24} (x_{12} + x_{32})}{x_{12}^2 x_{23}^2 x_{24}^3} \right],
\end{equation}
\begin{equation}
    \mathcal{I}_{1111} = y_{14} \oint_{x_1} \frac{dx}{\prod^4_{i = 1} (x-x_i)} + (1 \leftrightarrow 2) + (1 \leftrightarrow 3) = \frac{2 \pi i y_{14}}{x_{12} x_{13} x_{14}} + (1 \leftrightarrow 2) + (1 \leftrightarrow 3).
\end{equation}

Then, in $JT_{\mu}$ deformation, one might meet $f = (y - y_i)^n \ (n \geq 0, \ 1 \leq i \leq n)$ case, namely
\begin{equation}
    \mathcal{I}^f_{a_1 \cdots a_n} = \frac{1}{(a_1 -1) !} \cdots \frac{1}{(a_n - 1) !} \p_{x_1} \cdots \p_{x_n} \mathcal{I}^f_{x_1 \cdots x_n}.
\end{equation}
The new quantity $\mathcal{I}^f_{x_1 \cdots x_n}$ is introduced for simplicity
\begin{equation}
    \mathcal{I}^f_{x_1 \cdots x_n} = - \sum_{j = 1}^n \int^{y_i}_{y_j} \td y (y - y_i)^n \oint_{x_i} \frac{\td x }{\prod_{k = 1}^n (x - x_k)}.
\end{equation}
This can easily be calculated by using the residue theorem as
\begin{equation}
    \mathcal{I}^f_{x_1 \cdots x_n} = - \frac{2 \pi i}{n+1} \frac{y_{ji}^{n+1}}{\prod_{k (\neq i)} x_{ik}}.
\end{equation}
Here we will not show some specific examples since the procedure is the same as that of $f = 1$ case. The integrals in $T \overline{T}$ case can be easily re-derived by setting $n = 0$.

\bibliography{ref}

\end{document}